\def\draft{0}
\documentclass[11pt]{article}

\usepackage{amsfonts}
\usepackage{amsmath}
\usepackage{amssymb}
\usepackage{amsthm}
\usepackage{mathabx}
\usepackage{thmtools}
\usepackage{thm-restate}
\usepackage{bbm}
\usepackage{bm}
\usepackage{dsfont}
\usepackage{fullpage}
\usepackage{tikz}
\usepackage[most]{tcolorbox}
\usepackage{xcolor}

\newcommand*{\myfont}{\fontfamily{bch}\selectfont}
\DeclareTextFontCommand{\textmyfont}{\myfont}

\usepackage[linesnumbered,ruled,vlined]{algorithm2e} 
\SetKwComment{Comment}{/* }{ */}

\SetCommentSty{mycommfont}

\usepackage[hidelinks]{hyperref}
\usepackage[nameinlink]{cleveref}
\usepackage{url}            

\newtheorem{theorem}{Theorem}[section]

\newtheorem{lemma}[theorem]{Lemma}
\newtheorem{observation}[theorem]{Observation}

\newtheorem{definition}[theorem]{Definition}
\newtheorem{claim}[theorem]{Claim}

\newtheorem{remark}[theorem]{Remark}

\newcommand{\F}{\mathbb{F}}
\newcommand{\R}{\mathbb{R}}

\newcommand{\Boo}{\{0,1 \}}

\newcommand{\bigO}{\mathcal{O}}

\newcommand{\paren}[1]{\left( #1 \right)}
\newcommand{\brac}[1]{\left[ #1 \right]}
\newcommand{\set}[1]{\left\{ #1 \right\}}

\newcommand{\setcond}[2]{\left\{ #1 \;\middle\vert\; #2 \right\}}

\newcommand{\su}{\textnormal{sum}}

\DeclareMathOperator*{\E}{\mathbb{E}}

\newcommand{\cP}{\mathcal{P}}
\newcommand{\Z}{\mathbb{Z}}

\definecolor{thmcolor}{RGB}{235, 235, 235}
\definecolor{citecolor}{RGB}{1, 210, 56}

\newtcolorbox{algobox}{colback=lightgray!5!white,colframe=lightgray!75!black}
\newtcolorbox{thmbox}{colback=thmcolor!5!white,colframe=black!75!black}

\usepackage{setspace}
\usepackage{enumerate}
\usepackage[parfill]{parskip}
\usepackage{changepage}
\usepackage{framed}

\allowdisplaybreaks
\hypersetup{
	colorlinks,
	linkcolor={blue},
	citecolor={citecolor},
	urlcolor={blue}
}

\usepackage[
backend=biber,
style=alphabetic,
sorting=nyt,
backref=true,
maxcitenames = 8,
mincitenames = 5,
maxalphanames = 8,
minalphanames = 5,
maxnames = 10,
minnames = 5
]{biblatex}

\ifnum\draft=1
\newcommand{\anote}[1]{{\color{brown} [Amik: #1]}}
\newcommand{\mnote}[1]{{\color{red} [Madhu: #1]}}
\newcommand{\mpnote}[1]{{\color{pink} [Manaswi: #1]}}
\newcommand{\pnote}[1]{{\color{blue} [Prashanth: #1]}}
\newcommand{\snote}[1]{{\color{green} [Srikanth: #1]}}
\else  
\newcommand{\anote}[1]{}
\newcommand{\mnote}[1]{}
\newcommand{\mpnote}[1]{}
\newcommand{\pnote}[1]{}
\newcommand{\snote}[1]{}
\fi

\def\anon{0}
\date{July 3, 2025}

\begin{document} 
	\title{A Near-Optimal Polynomial Distance Lemma Over Boolean Slices}
 
    \if\anon1{}\else{    
    \author{Prashanth Amireddy\thanks{School of Engineering and Applied Sciences, Harvard University, Cambridge, Massachusetts, USA. Supported in part by a Simons Investigator Award and NSF Award CCF 2152413 to Madhu Sudan and a Simons Investigator Award to Salil Vadhan. Email: \texttt{pamireddy@g.harvard.edu}} \and
    Amik Raj Behera\thanks{Department of Computer Science, University of Copenhagen, Denmark. Supported by Srikanth Srinivasan's start-up grant from the University of Copenhagen. Email: \texttt{ambe@di.ku.dk} } \and
     Srikanth Srinivasan \thanks{Department of Computer Science, University of Copenhagen, Denmark. This work was funded by the European Research Council (ERC) under grant agreement no. 101125652 (ALBA).
     Email: \texttt{srsr@di.ku.dk} } \and 
     Madhu Sudan\thanks{School of Engineering and Applied Sciences, Harvard University, Cambridge, Massachusetts, USA. Supported in part by a Simons Investigator Award, NSF Award CCF 2152413 and AFOSR award FA9550-25-1-0112. Email: \texttt{madhu@cs.harvard.edu}}}
     }\fi

	\maketitle
        \pagenumbering{arabic}
\begin{abstract}
The celebrated Ore-DeMillo-Lipton-Schwartz-Zippel (ODLSZ) lemma asserts that $n$-variate non-zero polynomial functions of degree $d$ over a field $\mathbb{F}$, are non-zero over any ``grid'' (points of the form $S^n$ for finite subset $S \subseteq \mathbb{F}$) with probability at least $\max\{|S|^{-d/(|S|-1)},1-d/|S|\}$ over the choice of random point from the grid. In particular, over the Boolean cube ($S = \{0,1\} \subseteq \mathbb{F}$), the lemma asserts non-zero polynomials are non-zero with probability at least $2^{-d}$. In this work we extend the ODLSZ lemma optimally (up to lower-order terms) to ``Boolean slices'' i.e., points of Hamming weight exactly $k$. We show that non-zero polynomials on the slice are non-zero with probability $(t/n)^{d}(1 - o_{n}(1))$ where $t = \min\{k,n-k\}$ for every $d \leq k \leq (n-d)$. As with the ODLSZ lemma, our results extend to polynomials over Abelian groups. This bound is tight upto the error term as evidenced by multilinear monomials of degree $d$, and it is also the case that some corrective term is necessary. A particularly interesting case is the ``balanced slice'' ($k=n/2$) where our lemma asserts that non-zero polynomials are non-zero with roughly the same probability on the slice as on the whole cube.\\ 

\noindent
The behaviour of low-degree polynomials over Boolean slices has received much attention in recent years. However, the problem of proving a tight version of the ODLSZ lemma does not seem to have been considered before, except for a recent work of Amireddy, Behera, Paraashar, Srinivasan and Sudan (SODA 2025), who established a sub-optimal bound of approximately $((k/n)\cdot (1-(k/n)))^d$ using a proof similar to that of the standard ODLSZ lemma.\\

\noindent
While the statement of our result mimics that of the ODLSZ lemma, our proof is significantly more intricate and involves spectral reasoning which is employed to show that a natural way of embedding a copy of the Boolean cube inside a balanced Boolean slice is a good sampler.

\end{abstract} 

        \newpage 
        
\tableofcontents

        \newpage

\section{Introduction}\label{sec:intro}
The Ore-DeMillo-Lipton-Schwartz-Zippel (ODLSZ) \cite{ore1922hohere, DL78,Zippel79,Schwartz80} lemma captures the basic algebraic fact that a low-degree polynomial does not have many roots on a ``nice set'' of points. The standard nice set for this lemma is a grid $S^n$ (where $S$ is a finite subset of a field) and a version of this lemma states that no non-zero degree-$d$ polynomial can vanish on more that $d |S|^{n-1}$ points. This is easily seen to be tight: Take, for example, a univariate polynomial that has $d$ roots in $S$.

There also exist useful variants of this lemma for the case where $|S| < d.$ The example above shows that in general a degree-$d$ polynomial can vanish over all of $S^n$ and so some further condition is necessary. The most obvious condition is to simply force the polynomial to be non-zero on the grid $S^n.$ In the setting of the Boolean cube, i.e. $S = \{0,1\}$, which is the setting we study, this is equivalent to considering non-zero \emph{multilinear} polynomials of degree $d.$ In this setting (a variant of) the ODLSZ lemma states that a non-zero multilinear polynomial of degree $d$ is non-zero on at least $2^{n-d}$ points of $\{0,1\}^n.$ Again, this is tight: Take, e.g., a multilinear monomial of degree $
d.$

Though both these forms of the ODLSZ lemma are simple statements with easy inductive proofs, they have many different applications in the design of randomized algorithms \cite{Rabin-Vazirani-Matching}, probabilistically checkable proofs~\cite{BFLS, ALMSS}, pseudorandom constructions~\cite{DKSS, GRS-codingbook}, Boolean function analysis \cite{Nisan-Szegedy}, data communication~\cite{ABCO}, small-depth circuit lower bounds ~\cite{PaturiSaks, HRRY-threshold} and extremal combinatorics \cite{Saraf-Sudan}.

In this paper, we extend the ODLSZ lemma to a different nice set namely the \emph{Boolean slice}, which is an important subset of the Boolean cube $\{0,1\}^n$. For a parameter $k$, we use $\{0,1\}^n_k$ to denote the $k$th Boolean slice, i.e., the set of points in the cube of Hamming weight exactly $k$. The behavior of low-degree polynomials on Boolean slices has received quite a bit of attention recently with motivations from learning theory \cite{OW}, Boolean function analysis \cite{Wimmer, Filmus2014AnOB, FilmusKMW, Filmus-Ihringer19}, property testing~\cite{DDGKS17, KLMZ}, circuit lower bounds \cite{HRRY-threshold}, and local decoding algorithms \cite{ABPSS25-ECCC}. However, as far as we know, the natural question of finding a tight version of the ODLSZ lemma over Boolean slices has not been considered before. This is the question we address in this paper. 

More precisely, we consider the following question: 
\begin{center}
    \textcolor{red}{\textit{Given a polynomial $P$ of degree at most $d$ that does not vanish on $\{0,1\}^n_k$, how many zeroes can have $P$ have in this set?}}
\end{center}
This question makes sense when $d\leq t:= \min\{k,n-k\}$, since any function on $\{0,1\}^n_k$ can be expressed as a polynomial of degree $t.$

We give a near-optimal answer to this question for low-degree polynomials. More precisely, our main theorem is stated below. It holds for polynomials over any field and even in the case where the coefficients come from an \emph{Abelian group}\footnote{A multilinear polynomial over an Abelian group $G$ is of the form $\sum_{S\subseteq [n]} a_S \prod_{i\in S}x_i $ where $a_S\in G$ for each $S.$ Polynomials over such domains appear naturally in applications to circuit complexity~\cite{ACC-Torus-Learning} and additive combinatorics~\cite{TZ}.} (as is also true of the standard version of ODLSZ lemma over the Boolean cube).

\begin{thmbox}
\begin{restatable}[Main Theorem]{theorem}{mainthm}\label{thm:main}
There exists an absolute constant $\varepsilon>0$ so that the following holds. Fix an arbitrary Abelian group $G$ and a degree parameter $d \in \mathbb{N}$. For all natural numbers $n$ and $k$ such that $d\le k \leq n-d$, the following holds whenever $d \leq t^{\varepsilon}$ where $t = \min\{k,n-k\}$:\newline
For any degree-$d$ polynomial $P: \Boo^{n}_{k} \to G$ that does not vanish on $\{0,1\}^n_k$, we have
\begin{align*}
    \Pr_{\mathbf{x} \sim \Boo^{n}_{k}}[P(\mathbf{x}) \neq 0] \; \geq \; \bigg(\dfrac{t}{n}\bigg)^{d} \cdot \paren{1-\dfrac{1}{t^{\varepsilon}}}.
\end{align*}
\end{restatable}
\end{thmbox}
We prove \Cref{thm:main} in \Cref{ssec:general-proof}.\newline
At a high level, the uniform distribution on $\Boo^{n}_{k}$ is similar to the $(k/n)$-biased distribution, i.e. the distribution where each coordinate is independently $1$ with probability $k/n$. With slight modifications to the proof of ODLSZ lemma, one can show (see for example  \cite[Claim 6.8]{Dinur2017Agreement}) that the probability  of sampling a non-zero point from $(k/n)$-biased distribution is $(t/n)^{d}$, where $t = \min\set{k, n-k}$. The bound given by \Cref{thm:main} is equal to this bound up to small error terms.

\paragraph{Tightness.} The bound given is easily seen to be nearly tight using essentially the same example as in the case of the Boolean cube. For $k\leq n/2$, the monomial $x_1\cdots x_d$ is non-zero with probability approximately $(k/n)^d$, and there is a similar example for $k > n/2.$ Moreover, it is also possible to see that for certain $k$, an error term is required. For example, assume that $G$ is the finite field $\F_2$, $k=n/2$ and $d=1.$ Then the linear polynomial $x_1+x_2+1$ is non-zero with probability $1/2 - \Theta(1/n)$, implying that the monomial does not yield exactly the optimal bound. In the case that the degree $d=1$, we can improve the error parameter and show a bound of $t/n - 1/n$ (\Cref{thm:deg1} in \Cref{sec:deg1}).

\paragraph{Proof Techniques.} The standard proofs of the ODLSZ lemma follow a simple inductive strategy, using the obvious univariate case for both the base case and each inductive step of the argument. The recent work of Amireddy, Behera, Paraashar, Srinivasan and Sudan~\cite{ABPSS25-ECCC} used a similar idea to show the following sub-optimal bound. Unfortunately, it is not clear how to make the inductive strategy work for the slice to get a tight answer.\\

\begin{lemma}[Suboptimal distance lemma for slices]\label{lemma:suboptimal-DLSZ-slice}
\cite[Lemma 5.1.6]{ABPSS25-ECCC}. For every Abelian group $G$ and non-negative integers $d,k,n$ with $n \geq 1$ and $d\le k \leq n-d$ the following holds: 
For every degree-$d$ polynomial $P: \Boo^{n}_{k} \to G$ that does not vanish on $\{0,1\}^n_k$, we have
\begin{align*}
\Pr_{{\bf x}\sim \{0,1\}^n_k}[P(\mathbf{x}) \neq 0] \geq \binom{n-2d}{k-d}\bigg /{n\choose k}.
\end{align*}
In particular, for $k=n/2$ (for an even $n$), the above probability is at least $4^{-d}$.
\end{lemma}

A computation shows that for small $d$, the above implies that the fraction of points in $\{0,1\}^n_k$ where $P$ does not vanish is at least $((k/n)\cdot (1-(k/n)))^d$ (up to small error terms). When $k = n/2$, for example, this bound is $4^{-d}$ which is quadratically worse than \Cref{thm:main}.

To get the tight bound, we use a very different approach. We start with the above suboptimal bound, but combine it with spectral techniques, which we elaborate on next. Note that if the slice $k = n/2$, i.e. the balanced slice, then we get a bound of nearly $1/2^{d}$ which is essentially the same as the ODLSZ lemma over the Boolean cube $\Boo^{n}$ (\Cref{thm:DLSZ}). 
To prove this, the high-level idea is to consider the process of choosing a random subcube in the balanced Boolean slice $\Boo^{n}_{n/2}$ as follows: pair the $n$ coordinates into $n/2$ pairs uniformly at random, and in each such pair $\{x_i,x_j\}$, identify $x_i$ with the Boolean negation of $x_j$, i.e. $1-x_j.$ This gives us a random embedding of an $n/2$-dimensional cube in the slice $\Boo^{n}_{n/2}$ and the polynomial $P$ restricts to a degree-$d$ polynomial $Q$ on this subcube. If we could guarantee that $Q$ was always non-zero, then the standard ODLSZ lemma on the cube would give us the desired statement. Unfortunately, there are subcubes on which $P$ could be identically zero. The main technical lemma is to show that $Q$ is non-zero with high probability: intuitively, this is because the random process above is a good \emph{\underline{sampler}} of the balanced slice, i.e. the points in the randomly chosen subcube behave essentially like independent samples of the balanced slice.

Formally, the technical lemma is a statement about the approximate pairwise independence of two random points of the chosen subcube. We show (see \Cref{lemma:main-informal}) that the probability that two random points of this subcube lie in a set of density $\rho$ is roughly $\rho^2.$ This is done by analyzing a natural weighted graph $\Gamma$ on the balanced slice defined by the above sampling process. We show this via two arguments, depending on the regime of the degree parameter $d$.


For $d \leq C \log n$ for a constant $C > 0$, the main technical lemma follows from the use of the Expander mixing lemma \cite{Alon-Chung-EML}, which implies such a statement using bounds on the second-largest eigenvalue of the graph. To analyze the second-largest eigenvalue of $\Gamma$, we show that it can be embedded (as an induced subgraph) in a Cayley graph defined on the subgroup of $\F_2^n$ defined by points of even Hamming weight. The latter is easier to analyze using Boolean Fourier analysis, and an application of the eigenvalue interlacing theorem allows us to bound the eigenvalues of $\Gamma$. See \Cref{subsec:simple-proof-cayley} for more details. This easier case of the lemma is already interesting: for instance, it yields a different (arguably easier) proof of a junta theorem on the Boolean slice \cite{Filmus-Ihringer19}, analogous to a well-known theorem of Nisan and Szegedy \cite{Nisan-Szegedy}.

For $d = n^{\gamma}$ for a small constant $\gamma > 0$, we need to strengthen the guarantee of the sampler. To do so, we use the fact that the adjacency matrix for $\Gamma$ can be spectrally upper-bounded by another matrix\footnote{The technically accurate descriptor for this matrix is the `Noise operator in the Bose-Mesner algebra of the Johnson scheme.' See \Cref{sec:bal-slice} for details.} that satisfies a \emph{Hypercontractive inequality} (see \Cref{lemma:hypercontractive-log-Sobolev}). Intuitively, this is stronger than an eigenvalue bound, as the latter measures only the worst-case expansion of the underlying graph, while the former gives us stronger bounds on the expansion of smaller sets. Using this inequality alongside the Expander mixing lemma yields the desired pairwise independence. See \Cref{subsec:proof-johnson-scheme} for more details.



For imbalanced slices, i.e. $k \neq n/2$, we reduce to the balanced case via a random restriction idea (see \Cref{sec:all-slices}). The main conceptual idea is to obtain a basis for the space of polynomial functions on a slice. We note, essentially using an argument of Wilson \cite{WILSON90}, that for many distinct slices, the space of homogeneous multilinear monomials of degree $d$ forms a basis for the space of polynomials of degree $d$ on the slice (see \Cref{claim:homogeneous-spanning-positive-char}). Unlike other known bases for this space~\cite{Filmus2014AnOB}, this idea also works over fields of positive characteristic and even over cyclic groups of prime power order. For such `good' slices $k\leq n/2$, we reduce to a $2k$-dimensional cube via a random restriction (see \Cref{lemma:good-slice-positive-char}), which can easily be seen to leave the polynomial non-zero with probability $(2k/n)^d.$ Invoking the balanced case now concludes the lemma for the good slices.

Finally, to extend the main theorem to all slices, we note that for any slice $k$, there is a good slice not too `far away' (in the range $[k-\bigO(d),k]$) (see \Cref{lem:all-slices}). By setting a few variables at random to $1,$ we are able to reduce to a good slice.

\paragraph{Related Work.} As mentioned above, the study of low-degree polynomials over Boolean slices has received much attention in recent years. Closely related to this work is the work of Filmus~\cite{Filmus2014AnOB} that constructs a basis for the space of real-valued degree-$d$ polynomial functions over general Boolean slices. A recent result of Kalai, Lifshitz, Minzer and Ziegler \cite{KLMZ} constructs a \emph{dense model} for the balanced slice $\{0,1\}^n_{n/2}$ under the Gowers norm $U_d$; in particular, this implies that there is a subset $S$ of $\{0,1\}^n$ of constant density such that any polynomial of degree-$d$ has the same density over $S$ as it does over the balanced slice. In principle, both these works should be useful in order to prove a version of the ODLSZ lemma over Boolean slices. However, we note that each of these results is applicable over different domains ($\mathbb{R}$ or $\F_2$) while we prove a unified statement that holds over any Abelian group (and in particular over all fields).

\subsection{Applications of Optimal Distance Lemma}

To give some idea of the applicability of the ODLSZ lemma over the slice, we prove some variants of well-known theorems in combinatorics and Boolean function analysis.

\paragraph{Hyperplane covering.} Given a subset $S$ of the cube $\{0,1\}^n,$ we define the \emph{exact cover number of $S$}, denoted $\mathrm{ec}_n(S)$ to be the minimum number of hyperplanes (over some field $\F$) such that their union intersects $\{0,1\}^n$ exactly in the set $S$. A classical result of Alon and F\"{u}redi shows that for $S$ being the cube with a single point removed, $\mathrm{ec}_n(S) = n.$ This combinatorial result, which easily follows from with ODLSZ lemma over the cube, has seen many subsequent generalizations (e.g. \cite{Clifton-Huang,Sauermann-Wigderson-hyperplance-multiplicities, Bishnoi-hyperplance-subspace}).



Using just the sub-optimal version of the ODLSZ lemma (\Cref{lemma:suboptimal-DLSZ-slice}), we immediately get an optimal version of the hyperplane covering over a Boolean slice $\Boo^{n}_{k}$ with a missing point, instead of the whole Boolean cube $\Boo^{n}$. More precisely, for $S\subseteq \{0,1\}^n_k$, let $\mathrm{ec}_{n,k}(S)$ be the minimum number of hyperplanes (over some fixed field $\F$) such that their union intersects $\{0,1\}^n_k$ exactly in the set $S$. Following the idea of~\cite{Alon-Furedi}, we have the following.

\begin{thmbox}
\begin{restatable}{theorem}{hyperplaneslices}\label{thm:hyperplanes-slices}
Let $n,k$ be natural numbers with $k \in [n]$. Fix an arbitrary point $\mathbf{a} \in \Boo^{n}_{k}$. Then $\mathrm{ec}_{n,k}(\{0,1\}^n_k \setminus \{\mathbf{a}\}) = \min\{k,n-k\}.$
\end{restatable}
\end{thmbox}

\begin{proof}[Proof of \Cref{thm:hyperplanes-slices}]
Without loss of generality, we assume that $k \leq n/2$ and $\mathbf{a} = 1^{k} 0^{n-k}.$ Let $S$ denote $\{0,1\}^n_k \setminus \{\mathbf{a}\}.$

It is easy to see that $\mathrm{ec}_{n,k}(S) \leq k.$ The hyperplanes $H_{i} = \set{x_{i} = 0}$ for $i \in [k]$ cover exactly the points in $S.$

For the lower bound, assume for the sake of contradiction that there exists $m < k$ hyperplanes $H_{i} = \set{\ell_{i}(\mathbf{x}) = 0}$ (here $\ell_{i}(\mathbf{x})$ denotes a degree-$1$ polynomial and $i\in [m]$) 
covering exactly the points in $S.$
Then the polynomial $P(\mathbf{x}) := \prod_{i=1}^m \ell_i(\mathbf{x})$ is non-zero at exactly one point of $\{0,1\}^n_k$.

However, by \Cref{lemma:suboptimal-DLSZ-slice}, $P$ must be non-zero at at least $\binom{n-2m}{k-m}  > 1$ points (since $m < k \leq n/2$) of $\{0,1\}^n_k.$ Hence we arrive at a contradiction.
\end{proof}

\paragraph{A junta theorem for the slice.}
Nisan and Szegedy \cite{Nisan-Szegedy} showed that any Boolean function on $\Boo^{n}$ that has degree $d$ over $\mathbb{R}$ depends on $\bigO(d 2^{d})$ variables, i.e. it is a $\bigO(d 2^{d})$-junta. Chiarelli, Hatami, and Saks \cite{CHS-20} improved the bound to $\bigO(2^{d})$.  Filmus and Ihringer~\cite{Filmus-Ihringer19} extended this result to slices and showed that for a suitable range of $k$, any degree-$d$ (over $\mathbb{R}$) Boolean function on the slice $k$ is a restriction of a degree-$d$ function on $\{0,1\}^n$. Along with the result of~\cite{CHS-20}, this implies that such a function is an $\bigO(2^d)$-junta. While the results of~\cite{Nisan-Szegedy,CHS-20} are fairly elementary, the theorem of~\cite{Filmus-Ihringer19} is more involved, relying on the Log-Sobolev inequality and Hypercontractivity for the Boolean slice~\cite{LeeYau, DiaconisSaloffCoste}. 

Using \Cref{thm:main}, we show that we can avoid the use of advanced analytic techniques\footnote{We have two proofs of our main theorem. In the general case where $d$ can be as large as $n^{\Omega(1)}$, our proof also relies on hypercontractivity. However, in the case that $d \leq C \log n,$ which is also the main case of interest for junta theorems, our proof needs only basic Fourier analysis over the Boolean cube and the eigenvalue interlacing theorem.} in the proof of~\cite{Filmus-Ihringer19}, and give a direct proof (following the proof of~\cite{Nisan-Szegedy}) of the fact that any degree-$d$ Boolean function on the balanced slice $\{0,1\}^n_{n/2}$ depends on $\bigO(d2^d)$ variables (see \Cref{lem:improved-FI}). 
Plugging this into the proof of~\cite{Filmus-Ihringer19}, we can again recover the optimal bound of $\bigO(2^d).$ More details can be found in \Cref{sec:junta}.

 \paragraph{Organization of the paper.} We provide basic definitions and other preliminaries in~\Cref{sec:preliminaries}. We give the proof details of the distance lemma over the balanced slice (i.e., $k=n/2$) in~\Cref{sec:bal-slice} and then use this to get the same over {\em all} slices in~\Cref{sec:all-slices}, thus finishing the proof of our main theorem~\Cref{thm:main}. In~\Cref{sec:junta}, we present our alternate proof of the junta theorem over slices. Finally, we obtain an improvement of our main theorem for the case of linear functions (i.e., $d=1$) as~\Cref{thm:deg1} in~\Cref{sec:deg1}.

\section{Preliminaries}\label{sec:preliminaries}
\paragraph{Notations.}Let $(G, +)$ denote an Abelian group $G$ with addition as the binary operation. For any $g \in G$, let $-g$ denote the inverse of $g \in G$. For any $g \in G$ and integer $a \geq 0$, $a \cdot g$ (or simply $ag$) is the shorthand notation of $g + \ldots + g$ (taken $a$ times), and $-ag$ denotes $a\cdot (-g).$

For any $\mathbf{x} \in \Boo^{n}$, $|\mathbf{x}|$ denotes the Hamming weight of $\mathbf{x}$. For any $\mathbf{x}, \mathbf{y} \in \Boo^{n}$, let $\Delta(\mathbf{x},\mathbf{y})$ denote the Hamming distance between $\mathbf{x}$ and $\mathbf{y}$, i.e. $\Delta(\mathbf{x}, \mathbf{y}) = |\setcond{i \in [n]}{x_{i} \neq y_{i}}|$. For natural numbers $n$ and $k \leq n$, let $\Boo^{n}_{k}$ denote the subset of strings in $\Boo^{n}$ of Hamming weight exactly $k$.\newline
We denote the set of functions $f :\Boo^{n} \to G$ that can be expressed as a multilinear polynomial of degree $d$, with the coefficients being in $G$ by $\mathcal{P}_{d}(n,G)$. We also consider functions $f: \Boo^{n}_{k} \to G$. We denote the set of functions on $\Boo^{n}_{k}$ that can be expressed as a multilinear polynomial of degree $d$ with the coefficients in $G$ by $\mathcal{P}_{d}(n,k,G)$. We will simply write $\mathcal{P}_{d}(n,k)$ when $G$ is clear from the context.


For any natural numbers $n$ and $k \leq n$, $U_{n}$ denotes the uniform distribution on $\Boo^{n}$ and $U_{n,k}$ denotes the uniform distribution on $\Boo^{n}_{k}$. For a growing parameter $n$, $o_{n}(1)$ denotes a function that goes to $0$ as $n$ grows large.

\paragraph{Basic Tools.}We start with the standard ODLSZ lemma over the Boolean cube.\\

\begin{theorem}[ODLSZ lemma over $\{0,1\}^n$]
\label{thm:DLSZ}
Let $G$ be any Abelian group and let $P\in \mathcal{P}_d(n,G)$ be any non-zero polynomial. Then 
\[
\Pr_{\mathbf{x}\sim U_n} \left[P(\mathbf{x})\neq 0\right]\geq \frac{1}{2^d}.
\]
\end{theorem}


Another important tool we require is Lucas's theorem which allows us to compute binomial coefficients modulo a prime $p$.\\

\begin{lemma}[Lucas's Theorem~\cite{Lucas}]\label{lemma:lucas}
 Let $p$ be a prime number and $a$ and $b$ be any two natural numbers. Denote $a$ and $b$ in their unique $p$-ary representations as:
\begin{align*}
    a = \sum_{i = 0}^{\ell - 1} a_{i} p^{i}, \quad \quad b = \sum_{i = 0}^{\ell - 1} b_{i} p^{i}, \quad \quad  a_{i}, b_{i} \in \{0,1,\ldots, p-1\}
\end{align*}
Then,
\begin{align*}
    \displaystyle\binom{a}{b} \; \equiv \; \prod_{i = 0}^{\ell - 1} \displaystyle\binom{a_{i}}{b_{i}} \; \mod{p},
\end{align*}
where we define $\binom{x}{y}$ to be $0$ if $x < y$.
\end{lemma}

\noindent

We will need the following standard facts about expanders and Cayley graphs. We refer the reader to the survey~\cite{HLW} for more details.\\

\begin{definition}[Weighted Cayley Graph]
    Let $(G,+)$ be a finite Abelian group and $w:G \to \R^{\ge 0}$ be a {\em weight} function (we refer to the elements of non-zero weight as {\em generators}). We say that a weighted graph $\Gamma = \Gamma(G,w)$ defined as follows is a {\em weighted Cayley graph} over $G$.
    \begin{itemize}
        \item The vertices of $\Gamma$ are the elements of $G$.
        \item For every $g,g'\in G$, we add an edge $(g,g+g')$ with weight $w(g')$ to $\Gamma$.
    \end{itemize}
\end{definition}

The following lemma gives us a way of computing the eigenvalues of the adjacency matrix of weighted Cayley graphs over Abelian groups. \\

\begin{lemma}[Eigenvalues of Cayley graphs, see e.g.~\cite{HLW}]\label{lem:cayley-evals}
    Let $\Gamma = \Gamma(G,w)$ be a weighted Cayley graph over a finite Abelian group $G$, where $w:G \to \mathbb{R}^{\ge 0}$ is the corresponding weight function. Let $\chi:G \to \mathbb{C}^\times$ be an arbitrary group homomorphism (which we will refer to as a {\em character}). Then, $\chi$ is an eigenvector of the adjacency matrix of $\Gamma$ with eigenvalue equal to $\sum_{g\in G}w(g) \chi(G)$. 
\end{lemma}

\noindent
Following is a consequence of the \emph{expander mixing lemma}.\\

\begin{lemma}[Expander mixing lemma, see e.g.~\cite{HLW}~Lemma 2.5]\label{fact:eml}
    For a symmetric random walk matrix $W$ over vertices $V$ and every subset $S\subseteq V$, it holds that
    $$ \Pr_{u\sim V, v\sim N(u)}[u\in S\text{~and~}v\in S]\le \bigg(\frac{|S|}{|V|}\bigg)^2 + \mu(W)\cdot \frac{|S|}{|V|},$$ where $N(u)$ denotes the distribution over $V$ corresponding to taking a step from $u$ according to $W$ (i.e., the $u$-th row of $W$).
\end{lemma}

\section{Distance Lemma for the Balanced Slice}\label{sec:bal-slice}
In this section, we state the main technical lemma of our proof for \Cref{thm:main}. It is a statement on the ``expansion'' property of a graph $\Boo^{n}_{n/2}$, where the edge weights are given by a random process. We start by describing a random process that maps a string in $\Boo^{n}_{n/2}$ to a string in $\Boo^{n}_{n/2}$. In this section, we will always assume that $n$ is an even number.\\

\begin{definition}[The map $\Gamma$]
Let $\mathbf{a} \in \Boo^{n/2}$ and $\mathbf{u} \in \Boo^{n}_{n/2}$. Let $\mathbf{u}^{-1}\set{0}$ denote the set of coordinates where $\mathbf{u}$ is $0$, i.e. $\mathbf{u}^{-1}\set{0} \, = \, \setcond{i \in [n]}{u_{i} = 0}$. Similarly we have $\mathbf{u}^{-1}\set{1}$. let $\mathbf{u}^{-1}\set{1} = \set{i_{1},\ldots,i_{n/2}}$.\newline
For any perfect matching $\mathcal{M}$ between $\mathbf{u}^{-1}\set{0}$ and $\mathbf{u}^{-1}\set{1}$ ($\mathcal{M}$ is a bijection between these two sets), the function $\Gamma(\mathbf{u}, (\mathcal{M}, \mathbf{a}))$ is a balanced string $\mathbf{v} \in \Boo^{n}_{n/2}$ defined as follows:\newline
For every $k \in [n/2]$, $v_{i_{k}} = u_{i_{k}} \oplus a_{k}$ and $v_{\mathcal{M}(i_{k})} = u_{\mathcal{M}(i_{k})} \oplus a_{k}$.
\end{definition}

\noindent
In simple words, for every matching between the $0$-coordinates and $1$-coordinates of $\mathbf{u}$ and a string $\mathbf{a} \in \Boo^{n/2}$, we get a new balanced string $\mathbf{v}$ by flipping the endpoints of a subset of matching edges. Here the subset of matching edges whose endpoints are flipped is given by the string $\mathbf{a}$. Following is an example for $n = 8$.

\paragraph{Example.}Let $\mathbf{u} = 10101010$. Here $\mathbf{u}^{-1}\set{0} = \set{2,4,6,8}$ and $\mathbf{u}^{-1}\set{1} = \set{1,3,5,7}$. Let $\mathcal{M} = ((2,3),(6,1),(4,5),(8,7))$ and $\mathbf{a} = 0110$. Then $\Gamma(\mathbf{u}, (\mathcal{M}, \mathbf{a})) = \mathbf{v} = 00110110$ (endpoints of the $2^{nd}$ matching edge $(6,1)$ and the $3^{rd}$ matching edge $(4,5)$ are flipped).

\paragraph{}Next, we define a weighted graph on all the balanced strings with weights representing the probability of going from one balanced string to another for a random matching $\mathcal{M}$ and a random string $\mathbf{a}$ (using the map $\Gamma$).\newline
Let $n' = |\binom{n}{n/2}|$ denote the cardinality of the set of balanced strings $\Boo^{n}_{n/2}$. Let $G$ denote a weighted complete graph on $n'$ vertices, where the vertices denote strings in $\Boo^{n}_{n/2}$. For any two distinct balanced strings $\mathbf{u}, \mathbf{v} \in \Boo^{n}_{n/2}$, the weight of the edge $(\mathbf{u}, \mathbf{v})$, denoted by $w(\mathbf{u}, \mathbf{v})$ is:
\begin{align*}
    w(\mathbf{u}, \mathbf{v}) \; := \; \Pr_{\mathcal{M}, \mathbf{a}}[\Gamma(\mathbf{u}, (\mathcal{M}, \mathbf{a})) = \mathbf{v}],
\end{align*}
where the probability is over the choice of a random perfect labeled matching $\mathcal{M}$ between $\mathbf{u}^{-1}\set{0}$ and $\mathbf{u}^{-1}\set{1}$, and a uniformly random string $\mathbf{a} \in \Boo^{n/2}$. For every balanced string $\mathbf{u} \in \Boo^{n}_{n/2}$, we will denote by $W(\mathbf{u})$ the distribution on $\Boo^{n}_{n/2}$ where the probability of sampling $\mathbf{v}$ is equal to $w(\mathbf{u}, \mathbf{v})$. Let $W \in \mathbb{R}^{n' \times n'}$ denote the weighted adjacency matrix of $G$, i.e. 
\begin{align*}
    W[\mathbf{u}, \mathbf{v}] \, = \, w(\mathbf{u}, \mathbf{v}), \quad \quad \text{ for all } \mathbf{u}, \mathbf{v} \in \Boo^{n}_{n/2}
\end{align*}

\noindent
We are now ready to state the main technical lemma of our proof. It roughly says that if we sample a random vertex (which is a random balanced string) and its neighbour in the above-mentioned graph, then the two balanced strings behave ``almost like pairwise-independent'' points. In other words, the above-mentioned graph is a good sampler for the balanced slice $\Boo^{n}_{n/2}$.\\

\begin{thmbox}
\begin{restatable}[Main Lemma]{lemma}{mainlemma}\label{lemma:main-informal}
There exists a constant $\varepsilon > 0$ for which the following holds. Let $G$ and $W$ be as mentioned above and let $S \subseteq \Boo^{n}_{n/2}$ be an arbitrary subset of vertices with $|S| \geq 4^{-d} \cdot \binom{n}{n/2}$. Let $\rho$ denote the density of the set $S$. Then,
\begin{align*}
    \Pr_{\substack{\mathbf{x} \sim U_{n,n/2} \\ \mathbf{y} \sim W(\mathbf{x})}}[\mathbf{x} \in S \; \text{ and } \; \mathbf{y} \in S] \; \leq \; \rho^{2} \cdot \paren{1 + \dfrac{1}{n^{\varepsilon}}}.
\end{align*}
\end{restatable}
\end{thmbox}

\noindent
We will give two proofs for \Cref{lemma:main-informal}, for two regimes of the degree $d$:
\begin{enumerate}
    \item For degree $d \leq C\log n$ for some absolute constant $C > 0$, we give a simple argument using the spectral expansion properties of Cayley graphs and the expander mixing lemma. We prove this in \Cref{subsec:simple-proof-cayley}.
    \item For degree $d \leq n^{\gamma}$ for some absolute constant $\gamma > 0$, we rely on the spectrum of Johnson association schemes and use hypercontractivity for slice functions. We prove this in \Cref{subsec:proof-johnson-scheme}.
\end{enumerate}

\noindent
We will also need a lower bound on the probability in \Cref{lemma:main-informal}. This will hold for all degree $d$. Combining the upper and lower bounds (i.e.,~\Cref{lemma:main-informal} and~\Cref{lem:lower-bd} gives the final bound: see~\Cref{sec:everything}.)\\

\begin{lemma}[The lower bound]\label{lem:lower-bd}
Let $G$ be the graph as mentioned above and fix a degree parameter $d \in \mathbb{N}$. Let $P(\mathbf{x}):\Boo^{n}_{n/2} \to \mathbb{R}$ be a non-zero polynomial on the balanced slice $\Boo^{n}_{n/2}$ with $\deg(P) \leq d$. If $S \subseteq \Boo^{n}_{n/2}$ denote the set of non-zeroes of $P(\mathbf{x})$, then,
\begin{align*}
     \Pr_{\substack{\mathbf{x} \sim U_{n,n/2} \\ \mathbf{y} \sim W(\mathbf{x})}}[\mathbf{x} \in S \; \text{ and } \; \mathbf{y} \in S] \; \geq \; \dfrac{|S|}{\binom{n}{n/2}} \cdot \dfrac{1}{2^{d}}.
\end{align*}
\end{lemma}
 \begin{proof}[Proof of \Cref{lem:lower-bd}]
Note that it is sufficient to show that
\begin{align*}
    \Pr_{\mathbf{y} \sim W(\mathbf{x})}[\mathbf{y} \in S \, | \, \mathbf{x} \in S] \; \geq \; \dfrac{1}{2^{d}}, \quad \quad \text{ for all } \; \mathbf{x} \in S.
\end{align*}
Fix an arbitrary point $\mathbf{u} \in S$ and fix an arbitrary matching $\mathcal{M}$ between $\mathbf{u}^{-1}\set{0}$ and $\mathbf{u}^{-1}\set{1}$. We will show that for $1/2^{d}$-fraction of $\mathbf{a} \in \Boo^{n/2}$, the string $\Gamma(\mathbf{u}, (\mathcal{M}, \mathbf{a})) \in S$.\newline

\noindent
Define the polynomial $Q(z_{1},\ldots,z_{n/2}) := P(\Gamma(\mathbf{u}, (\mathcal{M}, \mathbf{z})))$. Note that $\deg(Q) \leq \deg(P) \leq d$ and $Q(\mathbf{0}) = P(\mathbf{u}) \neq 0$. Now using the standard ODLSZ lemma (\Cref{thm:DLSZ}) on $Q(\mathbf{z})$, we get,
\begin{align*}
    \Pr_{{\bf z} \sim \Boo^{n/2}}[Q({\bf z}) \ne 0] \; \ge \; \dfrac{1}{2^{d}} \quad \Rightarrow \quad \Pr_{\mathbf{a} \sim \Boo^{n/2}}[\Gamma(\mathbf{u}, (\mathcal{M}, \mathbf{a})) \in S] \; \geq \; \dfrac{1}{2^{d}}
\end{align*}
Since the above lower bound holds for every matching $\mathcal{M}$ between $\mathbf{u}^{-1}\set{1}$ and $\mathbf{u}^{-1}\set{0}$, we have,
\begin{align*}
    \Pr_{\mathcal{M},\mathbf{a}}[\Gamma(\mathbf{u}, (\mathcal{M}, \mathbf{a})) \in S] \; \geq \; \dfrac{1}{2^{d}}
\end{align*}
Since the above lower bound holds for arbitrary choice of $\mathbf{u} \in S$, this completes the proof of \Cref{lem:lower-bd}.

\end{proof}

\paragraph{}Next, we observe that the random process mentioned above is ``$S_{n}$-invariant''\footnote{$S_{n}$ is the group of permutations on $n$ elements.}, i.e. the probabilities do not change even if we simultaneously permute the coordinates of $\mathbf{u}$ and $\mathbf{v}$ (using the same permutation for both of them).\\

\begin{observation}\label{obs:undirected-matrix}
For any $\mathbf{u}, \mathbf{v} \in \Boo^{n}_{n/2}$, the weight $w(\mathbf{u}, \mathbf{v})$ depends only on\footnote{Recall that $\Delta(\cdot, \cdot)$ represents the Hamming distance.} $\Delta(\mathbf{u}, \mathbf{v})$. We have,
\begin{align*}
    w({\bf u},{\bf v}) \; = \; \dfrac{\Delta!(n/2-\Delta)!}{2^{n/2} \cdot (n/2)!} \; = \; \dfrac{1}{2^{n/2} \cdot {n/2\choose \Delta}}, \quad \quad \text{ where } \, 2\Delta=\Delta({\bf u},{\bf v}) \in [0,n]
\end{align*}
To see the above probability, observe that the $\frac{1}{2^{n/2}}$ factor corresponds to sampling the right $\mathbf{a}$ and the $\frac{\Delta!(n/2-\Delta)!}{(n/2)!}$ factor corresponds to picking a matching $\mathcal{M}$ that results in the output ${\bf v}$).
\end{observation}

\noindent
Note that the above observation in particular implies that the weighted adjacency matrix $W$ is a real \underline{symmetric} matrix, and thus has real eigenvalues. Both of our proofs for \Cref{lemma:main-informal} will be based on upper bounding the eigenvalues of $W$.

\subsection{Simple Proof Using Cayley Graphs}\label{subsec:simple-proof-cayley}
In this section, we prove a version of \Cref{lemma:main-informal} using a simple and (mostly) self-contained argument. This version holds for degrees $d \leq C \log n$ for some absolute constant $C$.\\

Let $1=\mu_1 \ge \mu_2 \ge \dots \ge \mu_{n'}$ be the eigenvalues of $W$ and let $\mu(W)$ denote the second largest eigenvalue in absolute value, i.e. $\mu(W):=\max(|\mu_2|,|\mu_{n'}|)$. A small value of $\mu(W)$ suggests that the random walk represented by $W$ is ``expanding'' (see \Cref{fact:eml}). The main lemma of this subsection is the following, which shows that $\mu(W)$ is small, i.e. $W$ is a good expander. In the rest of the subsection, let $m = n/2$.\\

\begin{lemma}[$W$ is a good expander]\label{lem:w-exp}
Let $W$ denote the $n' \times n'$ matrix as described before. Then,
\begin{align*}
    \mu(W) \le {\bigO}\bigg(\frac{\log n}{\sqrt{n}}\bigg).
\end{align*}
\end{lemma}

We now prove \Cref{lem:w-exp}, i.e., we show that $W$ is a good expander. The idea of the proof is to show that one can turn $W$ into a (weighted) Cayley graph by adding additional edges and vertices and deduce that the original graph is an expander by using the expansion of the Cayley graph. In particular, we will show that $W$ is an induced subgraph of a Cayley graph and use the interlacing of eigenvalues to prove the expansion.

\begin{proof}[Proof of \Cref{lem:w-exp}]
    For the proof, we will assume that $m$ is even; the odd case is handled similarly. Let $\{0,1\}^n_{odd}$ and $\{0,1\}^n_{even}$ denote the sets of points in $\{0,1\}^n$ that are of odd Hamming weight and even Hamming weight respectively. We will now define the weighted Cayley graph.

     Let $V'=\{0,1\}^{n}_{even}$ (note that $\Boo^{n}_{n/2} \subseteq V'$ as $n$ is assumed to be even). Note that $V'$ is an Abelian group with addition defined by performing coordinate sums modulo 2 (in particular, we may identify \{0,1\} with $\F_2$). We shall define a weighted Cayley graph $W'$ over vertices $V'$ by specifying its generators (and their weights) as follows. The set of generators is $\mathcal{S} =\{0,1\}^n_{even}$ and a generator ${\bf x} \in \mathcal{S}$ has weight
     \begin{align*}
         w(\mathbf{x}) \; = \; \dfrac{1}{2^{m} \cdot \binom{m}{\Delta}}, \quad \text{ where } \; |\mathbf{x}| = 2 \Delta \quad \text{ and } \quad 0 \leq \Delta \leq m
     \end{align*}

     With the above definition for $V'$ and $W'$, we note that the induced subgraph of $W'$ when restricted to the balanced slice $\Boo^{n}_{n/2} \subseteq V'$, is identical to $W$. Hence, by applying the eigenvalue interlacing theorem, we have the following.
     \begin{claim}[Eigenvalue interlacing, see e.g.~\cite{horn1991topics}]
        Let $\mu'_1 \ge \mu'_2 \ge \dots \ge \mu'_{|V'|}$ be the eigenvalues of $W'$. Then $\mu_2 \le \mu_2'$ and $\mu'_{|V'|} \le \mu_{n'}$. Hence, $\mu(W) \le \max(|\mu'_2|,|\mu'_{|V'|}|)$.
     \end{claim}
     The above claim allows us to bound $\mu(W)$ by bounding the absolute values of the eigenvalues of $W'$ (except the largest). To do this, we will first fix an eigenbasis for $W'$. The characteristic vectors of the first $(n-1)$ variables forms such an eigenbasis (because if $\mathbf{x} \in V'$, then $x_{n}$ can be expressed as a $\F_2$-linear combination of $x_{1},\ldots,x_{n-1}$). That is, for $A \subseteq [n-1]$, the characteristic vector ${\chi}_A \in \R^{2^{n-1}}$ is defined as ${\chi}_A({\bf x}) := (-1)^{\sum_{i\in A}{x_i}}$ for ${\bf x}\in V'$. The corresponding eigenvalue of $\chi_A$ is denoted by $\mu'_A$ (with a slight abuse of notation), and by~\Cref{lem:cayley-evals}, is equal to
     $$\mu'_A = \sum_{{\bf y}\in \mathcal{S}}w({\bf y})\chi_A({\bf y}).$$

     It will be convenient to normalize the weights of the generators in $\mathcal{S}$ to make it a probability distribution. More formally, let $\mathcal{D}$ be the probability distribution over $\mathcal{S}$ where the probability of sampling a point ${\bf x} \in \mathcal{S}$ is equal to $w({\bf x})/\sum_{{\bf y}\in \mathcal{S}}w({\bf y})$. Thus, $\mu'_{\emptyset} = \sum_{{\bf y}\in \mathcal{S}}w({\bf y})$ and $\mu'_A/\mu'_{\emptyset} =\E_{{\bf x}\sim \mathcal{D}}[\chi_A({\bf x})]$.\\

     \noindent 
     We will now show that $0 < \mu'_\emptyset \le \bigO(\sqrt{m})$ and $|\E_{{\bf x}\sim \mathcal{D}}[\chi_A({\bf x})]|=|\mu'_A/\mu'_{\emptyset}| \le O\big(\frac{\log m}{m}\big)$ for all {\em non-empty} $A \subseteq [2m-1]$. This would in turn give that $\mu(W) \le \max_{A \ne \emptyset}(|\mu'_A|) \le {\bigO}\bigg(\frac{\sqrt{m}\log m}{m}\bigg) = {\bigO}(\log m/\sqrt{m})$, finishing the proof of~\Cref{lem:w-exp}.

    The proof of \Cref{clm:eig-1} follows from a simple counting argument and can be found in \Cref{app:claims-simple-proof}.

    \begin{claim}\label{clm:eig-1}
        $0<\mu'_\emptyset \le \bigO(\sqrt{m}).$
    \end{claim}

    It remains to show that $|\E_{{\bf x}\sim \mathcal{D}}[\chi_A({\bf x})]| \le \bigO\big(\frac{\log m}{m}\big)$ for all {\em non-empty} $A \subseteq [2m-1]$. Note that the distribution $\mathcal{D}$ has some symmetry in the sense that all the points of a given Hamming weight have the same probability. Furthermore, two given points, one of weight $2 \Delta$ and the other of weight $(2m-2 \Delta)$ also have the same probability mass (for arbitrary $0 \leq \Delta \leq m$). This leads to $\E_{{\bf x}\sim \mathcal{D}}[\chi_A({\bf x})]$ being equal to $0$ if $|A|=1$ or $2m-1$, and hence it suffices to focus on the regime $2\le |A| \le 2m-2$.

    We show the following concentration inequality for the distribution $\mathcal{D}$. The proof of \Cref{clm:D-conc} can be found in \Cref{app:claims-simple-proof}.

    \begin{claim}\label{clm:D-conc}
        $\Pr_{{\bf x} \sim \mathcal{D}}[||{\bf x}|-m| > \sqrt{50 m \log m}] \le \bigO(1/m^2)$.
    \end{claim}

     Assuming \Cref{clm:D-conc}, it suffices to show that $\E_{{\bf x}\sim \mathcal{D}}[\chi_A({\bf x}) ~\mid~|{\bf x}| \in m\pm \sqrt{50 m \log m}] \le \bigO\big(\frac{\log m}{m}\big)$ to conclude that $\E_{{\bf x}\sim \mathcal{D}}[\chi_A({\bf x})] \le \bigO\big(\frac{\log m}{m}\big)$. We will show that this holds even conditioned on  $|{\bf x}|=2 \Delta$ for every $\Delta \in (m\pm \sqrt{50m\log m})/2$. However, recall that $\mathcal{D}$ is uniform when restricted to $\{0,1\}^{2m}_{2 \Delta}$. Therefore, we can equivalently upper bound the quantity $|\E_{{\bf x}\in \{0,1\}^{2m}_{2\Delta}}[\chi_A({\bf x})]|$ to conclude the proof. Now, we note that $\E_{{\bf x}\in \{0,1\}^{2m}_{2\Delta}}[\chi_A({\bf x})] = \E_{B \sim {[2m]\choose |A|}}[\chi_B({\bf c})]$, where ${\bf c}$ is an arbitrary point in $\{0,1\}^{2m}_{2\Delta}$ (we will fix it to be $0^{2m-2\Delta}1^{2\Delta}$). Hence, it suffices to show that 
     \begin{align}\label{eqn:exp-B}
        \bigg|\E_{B \sim {[2m] \choose k}}[\chi_B({\bf c})]\bigg| \; \le \; \bigO\bigg(\frac{\log m}{m}\bigg),
     \end{align}
     for every $2 \leq k \leq (2m-2)$ (since we assumed that $2\le |A| \le 2m-2$). We may further assume that $k\le m$ without loss of generality, as $\chi_B({\bf c}) = \chi_{\overline{B}}({\bf c})$.\\

     \noindent    
    To help with the analysis, we will choose $B\sim {[2m] \choose k}$ by first choosing a subset $C$ of $[2m]$ of size $(k-2)$ (which is non-negative) uniformly at random and then choosing two elements $b_1\ne b_2$ from $\overline{C}=[2m]\setminus C$ uniformly at random.
    
    For a subset $C \subseteq [2m]$, we will use the notation $\text{wt}(C)$ to denote the number of 1's in ${\bf c}$ when restricted to the coordinates indexed by $C$. Let $p:=\text{wt}([2m])/(2m)$ denote the fractional Hamming weight of ${\bf c}$. We say that a subset $C\subseteq [2m]$ is {\em good}, if $\big| |\overline{C}| - 2\text{wt}(\overline{C}) \big| \le \sqrt{{2000 \log m}}$. We claim that $\Pr_{C \sim {[2m] \choose k-2}}[C\text{~is not good}] \le \bigO(1/m^2)$. This essentially follows from standard tail bounds for the hypergeometric distribution but needs some care to handle small $k$. We divide this analysis into two cases.

    \begin{itemize}
        \item {\bf Case 1: $k \le \sqrt{50 m \log m}$.} We note that $|\overline{C}|=2m-k+2 \in 2m \pm \sqrt{50 m \log m}$, and similarly $\text{wt}(\overline{C}) \in \text{wt}([2m]) \pm \sqrt{50m\log m} \subseteq m \pm 2\sqrt{50m\log m}$. Using these bounds, it follows that $\big| |\overline{C}|-2\text{wt}(\overline{C}) \big| \le \sqrt{{2000\log m}}$ for sufficiently large $m$ (for every choice of $C\in {[2m] \choose k-2}$). Hence all choices of $C$ are good in this case.
        
        \item {\bf Case 2: $k > \sqrt{50 m \log m}$.} We note that $\text{wt}({C})$ is distributed according to a hypergeometric distribution --- it corresponds to the number of successes in $k-2$ draws with replacement from a population of $2m$ total states and $\text{wt}([2m])=2d$ success states. Using a standard tail bound~\cite{hoeffding1994probability}, we obtain that $\Pr_{C\sim {[2m]\choose k-2}}\big[|\frac{\text{wt}(C)}{k-2} - p|>\sqrt{\frac{4\log k}{k}}\big]= \bigO(1/k^4)= \bigO(1/m^2)$. Using $\big|p-\frac{1}{2}\big| \le \sqrt{\frac{50\log m}{m}}$, we thus get that $||C|-2\text{wt}(C)| \le 4\sqrt{50m\log m}$ with probability at least $1-\bigO(1/m^2)$. Because $2m=|C|+\overline{C}$ and $2d=\text{wt}(C)+\text{wt}(\overline{C})$, with probability at least $1-\bigO(1/m^2)$, we have $||\overline{C}|-2\text{wt}(\overline{C})| \le 6\sqrt{50m\log m}$, i.e., $C$ is good.
    \end{itemize}

    We now show that conditioned on $C$ being good, the expectation of $\chi_B({\bf c})$ is upper bounded by $\bigO\big(\frac{\log m}{m}\big)$ in absolute value. This would then prove~\eqref{eqn:exp-B}.  For ease of notation, let $n_0=|\overline{C}|=2m-k+2 \ge m$ and $d_0=\text{wt}(\overline{C}) \in \frac{n_0}{2} \pm \Theta(\sqrt{n_0\log n_0})$ (since $C$ is good).  We now note that $\chi_B({\bf c}) = \chi_{C}({\bf c})\cdot \chi_{\{b_1,b_2\}}({\bf c})$, so it suffices to bound $|\E_{b_1,b_2}[(-1)^{c_{b_1}+c_{b_2}}]|$. The idea now is that $c_{b_1}$ and $c_{b_2}$ almost behave like two independent draws, so the expectation is {\em roughly} the square of $|\E_b[(-1)^{c_b}]|$ for a uniformly random coordinate $b\in \overline{C}$, which is equal to $\big|\frac{n_0-2d_0}{n_0}\big| \le O(\sqrt{\frac{\log n_0}{n_0}}) \le O(\sqrt{\frac{\log m}{m}})$ as $n_0 \ge m$. More precisely, we have the following:
    \begin{align*}
        |\E_{b_1,b_2}[(-1)^{c_{b_1}+c_{b_2}}]| & = \frac{|{d_0\choose 2}+{n_0-d_0 \choose 2}-d_0(n_0-d_0)|}{{n_0\choose 2}}\\
        & = \frac{|(n_0-2d_0)^2-n_0|}{n_0(n_0-1)}\\
        & \le \bigO\bigg(\frac{\log n_0}{n_0}\bigg)\tag{as $d_0 \in n_0/2 \pm \Theta(\sqrt{n_0\log n_0})$} \le \bigO\bigg(\frac{\log n}{n}\bigg).
    \end{align*}
    
    This finishes the proof of~\Cref{lem:w-exp}.
\end{proof}


\begin{proof}[Proof of \Cref{lemma:main-informal} for $d \leq C \log n$] 
We use \Cref{lem:w-exp} to prove an upper bound for \Cref{lemma:main-informal} that holds for all $d \leq C \log n$ for an absolute constant $C > 0$. Using the expander mixing lemma (\Cref{fact:eml}), we obtain
$$\Pr_{{\bf x}\sim U_{n,n/2}, \\ {\bf y}\sim N({\bf x})}[{\bf x}\in S\text{~and~}{\bf y}\in S] \le \rho^2 + \rho \cdot \bigO\paren{\frac{\log n}{\sqrt{n}}},$$ where $S$ denotes the non-zeroes of $P$ in $\{0,1\}^n_{n/2}$ and $\rho=|S|/{n \choose n/2}$. Thus assuming $d\le C\log n$ for small enough constant $C$ and using $\rho \ge 4^{-d}$ (\Cref{lemma:suboptimal-DLSZ-slice}), we get that the above probability is at most $\rho^2(1+1/n^\varepsilon)$ for sufficiently small constant $\varepsilon$. Hence, this finishes the proof of~\Cref{lemma:main-informal} in the regime $d \le C\log n$.

\end{proof}

\subsection{Proof via Johnson Association Schemes}\label{subsec:proof-johnson-scheme}
In this section, we prove a stronger version of \Cref{lemma:main-informal}, i.e. we will give a tighter upper bound. This will hold for all $d \leq n^{\gamma}$ for some absolute constant $\gamma > 0$. In this subsection, we will always assume that $n$ is divisible by $2$. We will first give some preliminaries on \emph{Johnson association schemes} and functions on the balanced slice.\\

\subsubsection{Preliminaries for Johnson Association Schemes}
We now discuss that the matrix $W$ has some useful properties. Recall $n' = \binom{n}{n/2}$. Consider the set of $n' \times n'$ dimensional matrices satisfying the following property: For every entry $(\mathbf{u}, \mathbf{v})$, the entry only depends on the Hamming distance $\Delta(\mathbf{u}, \mathbf{v})$ as mentioned in \Cref{obs:undirected-matrix}. The set of all such matrices forms an algebra known as \emph{Bose-Mesner algebra of the $(n,n/2)$ Johnson association scheme}. Note that the matrices in this algebra are invariant under the action of the symmetric group $S_{n}$ on the coordinates, i.e., if we apply a permutation $\pi \in S_{n}$ on the coordinates of the rows and columns simultaneously, the matrix remains invariant. Filmus \cite{Filmus2014AnOB} showed that there exist vector spaces that form orthogonal eigenspaces for matrices in Bose-Mesner algebra of the $(n,n/2)$ Johnson association scheme with certain useful properties.\\

We equip the space of functions on the balanced slice $\Boo^{n}_{n/2}$ with an inner product. We consider the following inner product on functions on the balanced slice: For any two functions $f,g:\Boo^{n}_{n/2} \to \mathbb{R}$
\begin{align*}
    \langle f,g \rangle \; = \; \mathbb{E}_{\mathbf{x} \sim U_{n,n/2}}[f(\mathbf{x}) \cdot g(\mathbf{x})]
\end{align*}
The $p^{th}$ norm of a function $f: \Boo^{n}_{n/2} \to \mathbb{R}$ is defined as follows:
\begin{align*}
    \lVert f \rVert_{p} \; = \; \mathbb{E}_{\mathbf{x} \sim U_{n,n/2}}[|f(\mathbf{x})|^{p}]^{1/p} \; = \; \dfrac{1}{n'} \paren{\sum_{\mathbf{x} \in \Boo^{n}_{n/2}} |f(\mathbf{x})|^{p} }^{1/p}
\end{align*}

\paragraph{}Let us define a particular function on the balanced slice $\Boo^{n}_{n/2}$. For every $t \in \set{0,1,\ldots,n/2}$, let $f_{t}(\mathbf{x}): \Boo^{n}_{n/2} \to \mathbb{R}$ be the following function: For $t = 0$, $f_{t} = 1$ and for $t\geq 1$,
\begin{equation}\label{eqn:nice-eigenvector}
    f_{t}(x_{1},\ldots,x_{n}) \; = \; (x_{1} - x_{2}) \cdot (x_{3} - x_{4}) \cdots (x_{2t-1} - x_{2t})
\end{equation}
We are going to interpret $f_{t}(\mathbf{x})$ as a vector, i.e. for every $\bm{\alpha} \in \Boo^{n}_{n/2}$, the $\bm{\alpha}^{th}$ entry of the vector is $f_{t}(\bm{\alpha})$. We will use the following result from \cite{Filmus2014AnOB}.\\

\begin{lemma}\label{lemma:bose-mesner-eigenspaces}
\cite[Lemma 18]{Filmus2014AnOB}. There exists orthogonal\footnote{The orthogonality is with respect to the inner product defined in the previous paragraph.} vector spaces $\mathcal{V}_{n,0}, \ldots, \mathcal{V}_{n,n/2}$ for which the following holds. Let $W$ be any matrix in the Bose-Mesner algebra for the $(n,n/2)$ Johnson association scheme. Then,
\begin{enumerate}
    \item The spaces $\mathcal{V}_{n,0}, \ldots, \mathcal{V}_{n,n/2}$ are orthogonal eigenspaces for the matrix $W$ with corresponding eigenvalues (not necessarily distinct) $\lambda_{0}, \lambda_{1}, \ldots, \lambda_{n/2}$.
    \item For every integer $t \in \set{0,1,\ldots,n/2}$, the function $f_{t}(\mathbf{x})$, as described in \Cref{eqn:nice-eigenvector} lies in $\mathcal{V}_{n,t}$. In other words, $f_{t}$ is an eigenvector of $W$ in the $t^{th}$ eigenspace with eigenvalue $\lambda_{t}$.
\end{enumerate}
\end{lemma}

\paragraph{}For a function $f: \Boo^{n}_{n/2} \to \mathbb{R}$ and $0 \leq t \leq n/2$, let $f^{=t}$ denote the component of $f$ in the $t^{th}$ eigenspace $\mathcal{V}_{n,t}$ (see \Cref{lemma:bose-mesner-eigenspaces}).
We now define the following noise operator for functions on the slice.\\ 

\begin{definition}[Noise operator for functions on slice]
(See \cite[Section 2]{Filmus-Ihringer19}). For a parameter $\rho \in (0,1]$, the \emph{noise operator} $T_{\rho}$ maps functions on the slice $\Boo^{n}_{n/2}$ to functions on the slice $\Boo^{n}_{n/2}$, defined as:
\begin{align*}
    T_{\rho}f \; = \; \sum_{t=0}^{n/2} \; \rho^{t \paren{1 - \frac{t-1}{n}}} \, f^{=t}.
\end{align*}
\end{definition}

\paragraph{}Now we state a hypercontractive inequality for the noise operator $T_{\rho}$. Lee and Yau \cite{LeeYau} proved a log-Sobolev inequality which combined with a result of Diaconnis and Saloff-Coste \cite{DiaconisSaloffCoste} implies the following hypercontractive inequality. The following lemma is a simplified form of the more general inequality. For more details, refer to \cite[Section 2]{Filmus-Ihringer19} and \cite[Section 2]{Filmus2014AnOB}.\\

\noindent
\begin{lemma}[Hypercontractive inequality using log-Sobolev inequality]\label{lemma:hypercontractive-log-Sobolev}
Fix any multilinear polynomial $f$ on the balanced slice $\Boo^{n}_{n/2}$. Then there exists a constant $c$ such that for every $1 \leq p \leq q \leq \infty$ with $\frac{q-1}{p-1} \leq \exp{(c \log 1/\rho)}$,
\begin{align*}
    \lVert T_{\rho} f \rVert_{q} \; \leq \; \lVert f \rVert_{p}.
\end{align*}
\end{lemma}

\noindent
\begin{remark}
In \cite{Filmus2014AnOB} and  \cite{Filmus-Ihringer19}, the above lemma is stated in a slightly different way, so we take a moment to clarify that here. As mentioned in \cite[Page 2]{Filmus-Ihringer19} (see the line after Equation (2)), the noise operator $T_{\rho}$ is equivalent to the expected value after applying random $\mathrm{Po}(\frac{n-1}{2} \log (1/\rho))$ transpositions. Using the comment in \cite{Filmus2014AnOB} (after Definition 26), $T_{\rho}$ is equivalent to $H_{g(\rho)}$, where $H_{\cdot}$ is the Heat operator and $g(\rho) = \frac{n-1}{2} \log (1/\rho)$. We get \Cref{lemma:hypercontractive-log-Sobolev} by using \cite[Lemma 27]{Filmus2014AnOB} (the reader should be careful that the $\rho$ in that lemma is the \emph{log Sobolev constant} and is different from the $\rho$ in \Cref{lemma:hypercontractive-log-Sobolev}).
\end{remark}

\subsubsection{Proof of the Sampling Guarantee}
One of the key steps in our proof of \Cref{lemma:main-informal} will be \Cref{lemma:eigenvalue-bound} which gives an upper bound on the eigenvalues of the matrix $W$. It says that for small $t$, the eigenvalue $\lambda_{t}$ is quite small (roughly $1/n^{t}$) and for larger $t$, the eigenvalue is still exponentially small (roughly $1/2^{t}$). To upper bound the eigenvalues, the idea is to choose a suitable eigenvector and argue about its non-zero coordinates. In particular, we will work with the eigenvector $f_{t}(\mathbf{x})$ as stated in \Cref{eqn:nice-eigenvector} and \Cref{lemma:bose-mesner-eigenspaces}.

\noindent
\begin{lemma}\label{lemma:eigenvalue-bound}
Let $\lambda_{0},\lambda_{1},\ldots,\lambda_{n/2}$ denote the eigenvalues of $W$ and let $\tau := n^{\delta}$ for a sufficiently small constant $\delta > 0$. Then,
\begin{itemize}
    \item For $1 \leq t \leq \tau$, $\lambda_{t} \leq 1/n^{\Omega(t)}$
    \item For $t > \tau$, $\lambda_{t} \leq 1/2^{n^{\Omega(1)}}$.
\end{itemize}
\end{lemma}

\noindent
Before going into the proof of \Cref{lemma:eigenvalue-bound}, we will define a property on bipartite matching. For a matching $\mathcal{M}$, if $(i,j) \in \mathcal{M}$, then we will use the notation $\mathcal{M}(i) = j$ and $\mathcal{M}(j) = i$.\\

\noindent
\begin{definition}[Good and self-good matching]
Consider a complete bipartite graph $K_{n/2,n/2}$ on vertex set $(L \bigcup R)$, where $L = \set{1,3,\ldots,n-1}$ and $R = \set{2,4,\ldots,n}$. We will refer to a matching $\mathcal{M}$ between $L$ and $R$ as \emph{\underline{$t$-good}} if the following holds:
\begin{align*}
    \mathcal{M}(2i-1) \in \set{2,4,\ldots,2t} \quad \text{ or } \quad \mathcal{M}(2i) \in \set{1,3,\ldots,2t-1}, \quad \quad \text{ for all } \, i \in [t]
\end{align*}
We will call a good matching $\mathcal{M}$ a \emph{\underline{$t$-self good}} matching if for every subset $T \subseteq [t]$ of size $t/2$, there exists $i \in T$ such that
\begin{align*}
    \mathcal{M}(2i-1) = 2i
\end{align*}
\end{definition}

\noindent
We will simply refer to matchings that are not $t$-good as \emph{$t$-bad} matchings. We will refer to matchings that are $t$-good but not $t$-self good as \emph{$t$ non-self good} matchings. In our proof of \Cref{lemma:eigenvalue-bound}, it will be useful to have an upper bound on the probability that a random matching $\mathcal{M}$ is a $t$-good matching or a $t$-self good matching. We upper bound these probabilities in the following claim. The proof is a straightforward counting argument with standard binomial estimations. We omit the proof here and it can be found in \Cref{app:good-matching}.\\

\noindent
\begin{restatable}[Upper bound on probability of (self) good matchings]{claim}{upperboundgoodmatching}\label{claim:upper-bound-good-matching}
Consider the complete bipartite graph $K_{n/2,n/2}$ on $L \bigcup R$ where $L = \set{1,3,\ldots,n-1}$ and $R = \set{2,4,\ldots,n}$. Let $\tau = n^{\delta}$ for a sufficiently small $\delta > 0$. Then,
\begin{align*}
    \Pr_{\mathcal{M}}[\mathcal{M} \; \text{ is a $t$-good matching}] \; \leq \; \dfrac{1}{n^{\Omega(t)}}, \quad \quad \text{ for all } \; t \leq \tau, 
\end{align*}
where the above probability is over the choice of a uniformly random matching $\mathcal{M}$. Also,
\begin{align*}
    \Pr_{\mathcal{M}}[\mathcal{M} \; \text{ is a $t$-self good matching}] \; \leq \; \dfrac{1}{n^{\Omega(t)}}, \quad \quad \text{ for all } \; t > \tau,
\end{align*}
where the above probability is over the choice of a uniformly random matching $\mathcal{M}$.
\end{restatable}

\paragraph{}Now we have all the essentials with us to prove the \Cref{lemma:eigenvalue-bound}. As mentioned earlier, the idea would be to fix a non-zero coordinate of the eigenvector $f_{t}(\mathbf{x})$ and consider that coordinate in $Wf_{t}$. 
\begin{proof}[Proof of \Cref{lemma:eigenvalue-bound}]
Recall that the rows and columns of the matrix $W$ are indexed by points of $\Boo^{n}_{n/2}$. Fix any particular $t \in \set{1,\ldots,n/2}$.  We know that $f_{t}$ is an eigenvector of the matrix $M$ with eigenvalue $\lambda_{t}$, i.e. for every $\mathbf{u} \in \Boo^{n}_{n/2}$, we have,
\begin{align*}
    (Wf_{t})[\mathbf{u}] \; = \; \lambda_{t} \cdot f_{t}(\mathbf{u}) \quad
    \Rightarrow \quad \mathbb{E}_{\mathbf{v} \sim W(\mathbf{u})}[f_{t}(\mathbf{v})] \; = \; \lambda_{t} \cdot f_{t}(\mathbf{u})
\end{align*}
Fix any $\mathbf{u}$ for which $f_{t}(\mathbf{u}) =1$. For convenience, let $\mathbf{u} = 1010\ldots10$. Then we have,
\begin{align*}
    \lambda_{t} \; = \; \mathbb{E}_{\mathbf{v} \sim W(\mathbf{u})}[f_{t}(\mathbf{v})]
\end{align*}

\paragraph{}We make the following observation about bad matchings.\\

\begin{observation}\label{obs:bad-matching}
Fix $\mathcal{M}$ to be any $t$-bad matching. For simplicity in notation, assume that $\mathcal{M}(1) \notin \set{2,4,\ldots,2t}$ and $\mathcal{M}(2) \notin \set{1,3,\ldots,2t-1}$. Let $\mathcal{M}(2) = (2j-1)$ for some $j > t$. Then we have $v_{1} = u_{1} \oplus a_{1}$ and $v_{2} = u_{2} \oplus a_{j}$. Since $a_{1}$ and $a_{j}$ are mutually independent, this implies that $v_{1}$ and $v_{2}$ are mutually independent too. This implies that the expected value of $(v_{1} - v_{2})$ over the random choice of $\mathbf{a}$ is $0$ (conditioned on a bad matching $\mathcal{M}$). Using the independence of bits in $\mathbf{a}$, we have,
\begin{gather*}
    \mathbb{E}_{\mathbf{v} \sim W(\mathbf{u})}[f_{t}(\mathbf{v}) \, | \, \mathcal{M} \text{ is $t$-bad}] \; =  \; \mathbb{E}_{(a_{1},a_{j})} \; [(v_{1} - v_{2})] \cdot \mathbb{E}_{a_{2},\ldots,a_{n/2} \setminus a_{j}}  \; \brac{\prod_{k>1}^{t} (v_{2k-1} - v_{2k})} \\
    \Rightarrow \mathbb{E}_{\mathbf{v} \sim W(\mathbf{u})}[f_{t}(\mathbf{v}) \, | \, \mathcal{M} \text{ is $t$-bad}] \; = \; 0
\end{gather*}
\end{observation}

\paragraph{}This gives us the following:
\begin{align*}
    \lambda_{t} \; = \; \mathbb{E}_{\mathbf{v} \sim W(\mathbf{u})}[f_{t}(\mathbf{v}) \, | \, \mathcal{M} \, \text{ is $t$-good}] \cdot \Pr_{\mathcal{M}}[\mathcal{M} \, \text{ is $t$-good}] \, + \, \mathbb{E}_{\mathbf{v} \sim W(\mathbf{u})}[f_{t}(\mathbf{v}) \, | \, \mathcal{M} \, \text{ is $t$-bad}] \cdot \Pr_{\mathcal{M}}[\mathcal{M} \, \text{ is $t$-bad}]
\end{align*}
From \Cref{obs:bad-matching}, we know that the expected value of $f_{t}(\mathbf{v})$ for $\mathbf{v} \in N(\mathbf{u})$ conditioned on a bad matching $\mathcal{M}$ is $0$.
This gives us
\begin{align*}
    \lambda_{t} \;  = \; \mathbb{E}_{\mathbf{v} \sim W(\mathbf{u})}[f_{d}(\mathbf{v}) \, | \, \mathcal{M} \, \text{ is $t$-good}] \cdot \Pr_{\mathcal{M}}[\mathcal{M} \, \text{ is $t$-good}]
\end{align*}

\paragraph{Case 1 - $t \leq \tau$:}Note that the expectation is over a random choice of $\mathbf{a} \sim \Boo^{n/2}$. For any $\mathbf{v} \in \Boo^{n}_{n/2}$, $f_{t}(\mathbf{v}) \in \set{-1,0,1}$. In other words, the absolute value of the expectation is at most $1$. Using this, we now have,
\begin{align*}
    \lambda_{t} \; \leq \; |\mathbb{E}_{\mathbf{v} \sim W(\mathbf{u})}[f_{t}(\mathbf{v}) \, | \, \mathcal{M} \, \text{ is $t$-good}]| \cdot \Pr_{\mathcal{M}}[\mathcal{M} \, \text{ is $t$-good}] \; \leq \; \Pr_{\mathcal{M}}[\mathcal{M} \, \text{ is $t$-good}]
\end{align*}
For $t \leq \tau$, \Cref{claim:upper-bound-good-matching} implies that $\lambda_{t} \leq 1/n^{\Omega(t)}$. This shows the first item of \Cref{lemma:eigenvalue-bound}.

\paragraph{Case 2 - $t > \tau$:}We make the following observation about non-self good matchings.\\

\begin{observation}\label{obs:self-good-matching}
Fix a $t$ non-self good matching $\mathcal{M}$. By definition of $t$-self good matchings, we know that there exists a set $T \subseteq [t]$ of size $t/2$ such that for every $i \in T$, $\mathcal{M}(2i-1) \neq 2i$. Assume without loss of generality that $T = [t/2]$ and $\mathcal{M}(2i) = (2j_{i}-1)$ for $i \in T$. Note that for every $i \in T$, if $a_{i} \neq a_{j_{i}}$, then $f_{t}(\mathbf{x}) = 0$. In other words, for $f_{t}$ to be non-zero, it is necessary that for every $i \in [t/2]$, $a_{i} = a_{j_{i}}$. This implies that $f_{t}$ is non-zero for at most $\frac{1}{2^{t/4}}$ choices of $\mathbf{a}$.
\end{observation}

\noindent
Using \Cref{obs:self-good-matching}, we get the following upper bound on the expected value conditioned on a $t$ non-self good matching $\mathcal{M}$:
\begin{align*}
     \mathbb{E}_{\mathbf{v} \sim W(\mathbf{u})}[f_{t}(\mathbf{v}) \, | \, \mathcal{M} \; \text{is $t$ non-self good}] \; \leq \; |\mathbb{E}_{\mathbf{v} \sim W(\mathbf{u})}[f_{t}(\mathbf{v}) \, | \, \mathcal{M} \; \text{is $t$ non-self good}]| \; \leq \; \dfrac{1}{2^{t/4}}
\end{align*}
Thus finally we have,
\begin{align*}
    \lambda_{t} \; = \;  \mathbb{E}_{\mathbf{v} \sim W(\mathbf{u})}[f_{t}(\mathbf{v}) \, | \, \mathcal{M} \, \text{ is $t$-good}] \cdot \Pr_{\mathcal{M}}[\mathcal{M} \, \text{ is $t$-good}] \\ \\
    = \; \mathbb{E}_{\mathbf{v} \sim W(\mathbf{u})}[f_{t}(\mathbf{v}) \, | \, \mathcal{M} \, \text{ is $t$-self good}] \cdot \Pr_{\mathcal{M}}[\mathcal{M} \, \text{ is $t$-self good}] \\ + \,  \mathbb{E}_{\mathbf{v} \sim W(\mathbf{u})}[f_{t}(\mathbf{v}) \, | \, \mathcal{M} \, \text{ is $t$ non-self good}] \cdot \Pr_{\mathcal{M}}[\mathcal{M} \, \text{ is $t$ non-self good}] \\ \\
    \leq \; \Pr_{\mathcal{M}}[\mathcal{M} \, \text{ is $t$-self good}] \, + \,  \mathbb{E}_{\mathbf{v} \sim W(\mathbf{u})}[f_{t}(\mathbf{v}) \, | \, \mathcal{M} \, \text{ is $t$ non-self good}] \\
    \Rightarrow \; \lambda_{t} \; \leq \;  \dfrac{1}{2^{t/4}} \, + \, \dfrac{1}{n^{\Omega(t)}} \; \leq \; \dfrac{1}{2^{n^{\Omega(1)}}}
\end{align*}
This shows the second item of \Cref{lemma:eigenvalue-bound} and completes the proof of \Cref{lemma:eigenvalue-bound}.

\end{proof}

We are now ready to prove our main lemma (\Cref{lemma:main-informal}) in the setting when $d \leq n^{\kappa}$. We recall the statement below.\\

\mainlemma*

\begin{proof}[Proof of \Cref{lemma:main-informal} for $d \leq n^{\gamma}$]
Let $\mathds{1}_{S}$ denote the $n'$-dimensional characteristic vector for the subset $S$. Then,
\begin{align*}
    \Pr_{\substack{\mathbf{x} \sim U_{n,n/2} \\ \mathbf{y} \sim W(\mathbf{x}) }}[\mathbf{x} \in S \; \text{ and } \; \mathbf{y} \in S] \; = \; \langle \mathds{1}_{S}, \; W \mathds{1}_{S} \rangle,
\end{align*}
where $\langle f,g \rangle = \mathbb{E}_{\mathbf{x} \sim U_{n,n/2}}[f(\mathbf{x}) g(\mathbf{x})]$. Let $\mathcal{V}_{n,0},\ldots,\mathcal{V}_{n,n/2}$ be the orthogonal basis for the space of functions on $\Boo^{n}_{n/2}$ as stated in \Cref{lemma:bose-mesner-eigenspaces}. Let $\mathds{1}_{S}^{=t}$ denote the component of $\mathds{1}_{S}$ in the $t^{th}$ eigenspace $\mathcal{V}_{n,t}$. We have,
\begin{align*}
    W \mathds{1}_{S} \; = \; \sum_{t=0}^{n/2+1} W \mathds{1}_{S}^{=t} \; = \; \sum_{t=0}^{n/2+1}\lambda_{t} \mathbf{1}_{S}^{=t},
\end{align*}
where for the final equality we used the fact that $\mathds{1}_{S}^{=t}$ is an eigenvector for $W$ (first item of \Cref{lemma:bose-mesner-eigenspaces}). Using this, we have,
\begin{align*}
    \langle \mathds{1}_{S}, \; W \mathds{1}_{S} \rangle \; = \; \left\langle \sum_{t=0}^{n/2+1} \mathds{1}_{S}^{=t}, \; \sum_{t=0}^{n/2+1} \lambda_{t} \mathds{1}_{S}^{=t} \right\rangle \; = \; \sum_{t=0}^{n/2+1} \lambda_{t} \lVert \mathds{1}_{S}^{=t} \rVert_{2}^{2},
\end{align*}
where for the final equality we used the orthogonality of $\mathcal{V}_{n,t}$'s.

\paragraph{}Using the orthogonality $\mathcal{V}_{n,t}$'s, we have,
\begin{align*}
    \lVert T_{\rho} \mathds{1}_{S} \rVert_{2}^{2} \; = \; \sum_{t=0}^{n/2+1} \; \rho^{2t \paren{1 - \frac{t-1}{n}}} \, \lVert \mathds{1}_{S}^{=t} \rVert_{2}^{2}
\end{align*}
Let $\tau := n^{\delta}$ for small enough $\delta > 0$ such that \Cref{lemma:eigenvalue-bound} holds. If we have $\rho = 1/n^{\delta'}$ for a small enough $\delta'$ depending on $\delta$, then for $t \leq \tau$, we have,
\begin{align*}
    \lambda_{t} \, \leq \, \rho^{2t} \, \leq \, \rho^{2t \paren{1 - \frac{t-1}{n}}}
\end{align*}
This implies that for $t \leq \tau$, we have,
\begin{equation}
    \sum_{t=0}^{\tau} \lambda_{t} \lVert \mathds{1}_{S}^{=t} \rVert_{2}^{2} \; \leq \; \sum_{t=0}^{\tau} \; \rho^{2t \paren{1 - \frac{t-1}{n}}} \, \lVert \mathds{1}_{S}^{=t} \rVert_{2}^{2} \; \leq \; \lVert T_{\rho} \mathds{1}_{S} \rVert_{2}^{2}
\end{equation}
For $t > \tau$, we have,
\begin{equation}
    \sum_{t=\tau+1}^{n/2+1} \lambda_{t} \lVert \mathds{1}_{S}^{=t} \rVert_{2}^{2} \; \leq \; \sum_{t=\tau}^{n/2+1} \dfrac{1}{2^{n^{\Omega(1)}}} \lVert \mathds{1}_{S}^{=t} \rVert_{2}^{2} \; \leq \; n \cdot \dfrac{1}{2^{n^{\Omega(1)}}},
\end{equation}
where we upper bounded $\lVert \mathds{1}_{S}^{=t} \rVert_{2}^{2}$ by $1$. Combining these two together, we get,
\begin{align*}
    \sum_{t=0}^{n/2+1} \lambda_{t} \lVert \mathds{1}_{S}^{=t} \rVert_{2}^{2} \; \leq \; \lVert T_{\rho} \mathds{1}_{S} \rVert_{2}^{2} \, + \, \dfrac{n}{2^{n^{\Omega(1)}}}
\end{align*}
Now applying the hypercontractivity theorem for the noise operator $T_{\rho}$, we get,
\begin{align*}
    \lVert T_{\rho} \mathds{1}_{S} \rVert_{2} \; \leq \; \lVert \mathds{1}_{S} \rVert_{p},
\end{align*}
where $1/(p-1) \leq \exp{(c \log 1/\rho)} = \exp{(c \delta' \log n)}$.
We also have that for any $p$, the norm $\lVert \mathds{1}_{S} \rVert_{p}$ is equal to $(|S|/n')^{1/p}$. Using this, we get,
\begin{align*}
     \lVert T_{\rho} \mathds{1}_{S} \rVert_{2}^{2} \; \leq \; \paren{\dfrac{|S|}{n'}}^{2/p} \; = \; \paren{\dfrac{|S|}{n'}}^{2 (1-1/n^{\kappa})},
\end{align*}
for some constant $\kappa$ depending on the constant $c$ from the hypercontractive inequality \Cref{lemma:hypercontractive-log-Sobolev} and $\rho$. Plugging this back in, we have,
\begin{gather*}
    \Pr_{\substack{\mathbf{x} \sim U_{n,n/2} \\ \mathbf{y} \sim W(\mathbf{x}) }}[\mathbf{x} \in S \; \text{ and } \; \mathbf{y} \in S] \; = \; \langle \mathds{1}_{S}, \; W \mathds{1}_{S} \rangle 
    = \; \sum_{t=0}^{n/2+1} \lambda_{t} \lVert \mathds{1}_{S}^{=t} \rVert_{2}^{2}\\ 
    \leq \; \paren{\dfrac{|S|}{n'}}^{2 (1-1/n^{\kappa})} \, + \, \dfrac{1}{2^{n^{\Omega(1)}}} = \paren{\dfrac{|S|}{n'}}^{2}\cdot \left(\paren{\dfrac{|S|}{n'}}^{-2/n^{\kappa}} + \frac{(n'/|S|)^2}{2^{n^{\Omega(1)}}}\right).
\end{gather*}
From the hypothesis of \Cref{lemma:main-informal}, we know that $|S|/n' \geq 4^{-d}$. Using this in the parenthetical term above, we get
\begin{align*}
    \Pr_{\substack{\mathbf{x} \sim U_{n,n/2} \\ \mathbf{y} \sim W(\mathbf{x}) }}[\mathbf{x} \in S \; \text{ and } \; \mathbf{y} \in S] \; \leq \; \paren{\dfrac{|S|}{n'}}^{2}\cdot 
    \left(4^{\bigO(d/n^\kappa)} + \frac{4^{\bigO(d)}}{2^{n^{\Omega(1)}}}
    \right) \;
    \leq \; \paren{\dfrac{|S|}{n'}}^{2} \cdot \paren{1 + \dfrac{1}{n^{\varepsilon}}}.
\end{align*}
for a small enough absolute constant $\varepsilon > 0$ as long as $d\leq n^\gamma$ for a small enough absolute constant $\gamma > 0.$ This concludes the proof of \Cref{lemma:main-informal}.
\end{proof}

\subsection{Putting Everything Together}\label{sec:everything}
We now use the above bounds to show that over the balanced slice (the set of points with Hamming weight $n/2$), we have the optimal distance lemma for low-degree polynomials. The main result of this section is the following lemma.




\begin{thmbox}
\begin{theorem}[Distance lemma over the balanced slice]\label{lemma:balanced-slice}
There exists an absolute constant $\varepsilon > 0$ so that the following holds. Fix an arbitrary Abelian group $G$ and fix a degree parameter $d \in \mathbb{N}$ where $d \leq n^{\varepsilon}$. For every even natural number $n$, and for every non-zero degree-$d$ polynomial $P(\mathbf{x}) \in \mathcal{P}_{d}(n,n/2, G)$,
\begin{align*}
    \Pr_{\mathbf{x} \sim U_{n,n/2}  }[P(\mathbf{x}) \neq 0] \; \geq \; \dfrac{1}{2^{d}} \cdot \paren{1 - \dfrac{1}{n^{\Omega(1)}} }
\end{align*}
\end{theorem}
\end{thmbox}

\begin{proof}[Proof of \Cref{lemma:balanced-slice}]
Letting $S \subseteq \{0,1\}^n_{n/2}$ denote the set of points on the balanced slice on which $P$ evaluates to a non-zero value. From \Cref{lemma:suboptimal-DLSZ-slice}, the set $S$ satisfies the density lower required in \Cref{lemma:main-informal}. Combining~\Cref{lemma:main-informal} and~\Cref{lem:lower-bd}, we obtain
\begin{align*}
    \dfrac{|S|}{\binom{n}{n/2}} \cdot \dfrac{1}{2^d} \le \Pr_{\substack{\mathbf{x} \sim U_{n,n/2} \\ \mathbf{y} \sim W(\mathbf{x})}}[\mathbf{x} \in S \; \text{ and } \; \mathbf{y} \in S] \le \paren{\dfrac{|S|}{\binom{n}{n/2}}}^2 \cdot \paren{1+\frac{1}{n^\varepsilon}},
\end{align*}
 where $\varepsilon$ is a sufficiently small constant. Hence, $|S|/\binom{n}{n/2} \; \ge \; \frac{1}{2^d}\cdot \paren{1-\frac{1}{n^{\Omega(1)}}}$.
\end{proof}

\section{Arbitrary Slices} \label{sec:all-slices}
In this section we prove a distance lemma for slices over arbitrary Abelian groups $G$. As discussed in the proof overview, the proof has three key steps:
\begin{enumerate}
    \item First, we prove a distance lemma over cyclic groups of prime power order for some fixed set of slices, which we refer to as ``good'' slices. We prove this in \Cref{lemma:good-slice-positive-char}.
    \item Secondly, we prove a distance lemma over cyclic groups of prime power order for any slice by reducing it to one of the good slices. We prove this in \Cref{lem:all-slices}.
    \item In the end, we show that any Abelian group can be assumed to be a finitely generated Abelian group. To prove the distance lemma for a finitely generated Abelian group, it suffices to have the distance lemma over cyclic groups of prime power order, which we prove in \Cref{lemma:DLSZ-slice-finite-fields}.
\end{enumerate}

We will start by proving a distance lemma for slices over cyclic groups of prime power order.

\subsection{Cyclic Groups of Prime Power Order}
In this subsection we will prove the distance lemma for slices over $\mathbb{Z}_{p^{\ell}}$ for some prime $p$ and natural number $\ell$. The main result of this subsection is the following lemma.

\begin{thmbox}
\begin{restatable}[Distance lemma for cyclic groups]{lemma}{cyclicgroups}\label{lemma:DLSZ-slice-finite-fields}
There exists an absolute constant $\varepsilon > 0$ so that the following holds. Fix a cyclic group $\mathbb{Z}_{q}$ where $q = p^{\ell}$ for some prime $p$ and a degree parameter $d \in \mathbb{N}$. For all natural numbers $n$ and $k$ such that $1+d\le k^\varepsilon$ and $k\leq n/2$, the following is true.\newline
For a non-zero degree-$d$ polynomial $P: \Boo^{n}_{k} \to \mathbb{Z}_{q}$, 
\begin{align*}
    \Pr_{\mathbf{x} \sim U_{n,k}}[P(\mathbf{x}) \neq 0] \; \geq \; \alpha^{d} \paren{1- \dfrac{1}{k^{\Omega(1)}}}, \quad \quad \text{ where } \alpha := \frac{k}{n}.
\end{align*}
\end{restatable}
\end{thmbox}

We will start with the definition of good slices. The idea is that these slices admit a nice basis which can be used to reduce the problem from a good slice to the balanced slice over a smaller dimensional cube, which we have proved in \Cref{lemma:balanced-slice}.\\

\begin{definition}[Good slices]\label{defn:good-slices}
Fix a prime $p$ and a degree parameter $d \in \mathbb{N}$.
An integer $k\geq d$ is said to be $(d,p)$-\textit{good} if the $p$-ary expansion of $k$ agrees with that of $d$ in all the digits up to the leading digit of $d$, and is greater than equal to $d$ in the leading digit. More formally, if $k = \sum_{j = 0}^m a_j p^j$ and $d = \sum_{j=0}^\ell b_j p^j$ with $a_j, b_j \in \{0,\ldots, p-1\},$ with $b_\ell > 0$, then $a_j = b_j$ for all $j < \ell$ and $a_\ell \geq b_\ell.$
\end{definition}

For a degree parameter $d$, let $\mathcal{H}_{d}$ denote the set of homogeneous monomials in $\set{x_{1},\ldots,x_{n}}$, i.e.
\begin{align*}
    \mathcal{H}_{d} \; = \; \setcond{\prod_{i \in T} x_{i}}{T \subseteq [n], |T| = d}
\end{align*}
Since every $x_{i} \in \Boo$, we only work with multilinear monomials. So for convenience, we will identify monomials with sets and vice-versa and for a set $T \subseteq [n]$, let $x^{T} := \prod_{i \in T} x_{i}$.\newline

\noindent
We will now show that if $k$ is $(d,p)$-good, then the set of degree-$d$ homogeneous multilinear monomials $\mathcal{H}_{d}$ form a `basis' for degree-$d$ polynomials on the $k^{th}$ slice in the following sense.\\

\begin{lemma}
    \label{lem:basis}
    Let $q$ be a power of prime $p$ as above. Fix $n,k,d$ such that $d\le k\leq n-d$ and assume that $k$ is $(d,p)$-good. Then, any function in $\mathcal{P}_d(n,k,\mathbb{Z}_q)$ can be written uniquely as a linear combination of the monomials in $\mathcal{H}_d.$
\end{lemma}

To prove the above lemma, we will first show that there are at least $q^{\binom{n}{d}}$ distinct degree-$d$ polynomial functions on the slice, and then that the monomials in $\mathcal{H}_d$ span all these functions. Since $|\mathcal{H}_{d}|$ is exactly $\binom{n}{d}$, these two statements immediately implies \Cref{lem:basis}. 

We start with the lower bound on $|\mathcal{P}_d(n,k,\mathbb{Z}_q)|.$ The proof is implicit in the work of Wilson \cite{WILSON90}. For the sake of completeness, we give a proof in \Cref{app:wilson-proof}.\\

\begin{restatable}[Number of degree-$d$ polynomials on the slice]{lemma}{degdslicedimension}\label{lem:dim-wilson}
\cite{WILSON90}. For every degree parameter $d \in \mathbb{N}$ and for every slice parameter $k$ such that $d\leq \min\{k,n-k\}$, the number of distinct degree-$d$ polynomial functions on $\Boo^{n}_{k}$ is at least $q^{\binom{n}{d}}$. 
\end{restatable}

We now show that $\mathcal{H}_{d}$ is a spanning set for $\mathcal{P}_d(n,k,\mathbb{Z}_q)$, which is our next claim. This proof is also inspired by~\cite{WILSON90} who proves this in the case when the polynomials have coefficients that are real numbers.

\noindent
\begin{claim}\label{claim:homogeneous-spanning-positive-char}
Fix a cyclic group $\mathbb{Z}_{q}$ where $q$ is a power of prime $p$. Fix a degree parameter $d \in \mathbb{N}$. For all natural numbers $n,k$ such that $d\le k \leq (n-d)$ and $k$ is $(d,p)$-good, the following holds.\newline
The set $\mathcal{H}_{d}$ of homogeneous degree-$d$ monomials is a spanning set for $\mathcal{P}_{d}(n,k,\mathbb{Z}_{q})$.
\end{claim}

\paragraph{Proof Idea:}Every monomial of degree strictly less than $d$ can be expressed as a linear combination of monomials of degree exactly equal to $d$ using the fact that we are working over the slice $\Boo^{n}_{k}$. As we are working over a group of prime power order, we have to be careful about the coefficients arising while expressing lower degree monomials using homogeneous monomials. We will use that $k$ is $(d,p)$-good and Lucas's theorem (\Cref{lemma:lucas}) in a crucial way to argue about the coefficients.\\

\begin{proof}[Proof of \Cref{claim:homogeneous-spanning-positive-char}]
Consider any monomial $\mathfrak{m}$ of degree $0\leq i < d$. For simplicity in notations, assume without loss of generality that $\mathfrak{m} = x_{1} x_{2} \cdots x_{i}$. Let $\mathcal{H}_{d}|_{\mathfrak{m}}$ denote the subset of $\mathcal{H}_{d}$ which is divisible by $\mathfrak{m}$, i.e.
\begin{align*}
    \mathcal{H}_{d}|_{\mathfrak{m}} \; := \; \setcond{x^{T} \in \mathcal{H}_{d}}{\set{1,\ldots,i} \subset T}
\end{align*}
Since every monomial in $\mathcal{H}_{d}|_{\mathfrak{m}}$ is divisible by $\mathfrak{m}$, it is easy to verify that over the slice $\Boo^{n}_{k}$, the following identity holds:
\begin{align*}
    \sum_{x^{T} \in \mathcal{H}_{d}|_{\mathfrak{m}}} x^{T} 
    \; =  \; \mathfrak{m} \sum_{\substack{T' \subseteq [n] \setminus [i] \\ |T'| = d - i}} x^{T'} \; = \; \; \mathfrak{m} \cdot \binom{k - i}{d - i} 
\end{align*} 
To write the monomial $\mathfrak{m}$ as a linear combination of monomials in $\mathcal{H}_{d}$, we need the integer $\binom{k-i}{d-i}$ to be invertible in the ring $\mathbb{Z}_q$. This happens, exactly when $\binom{k-i}{d-i}$ is non-zero modulo the prime $p$.\\

\noindent
Using Lucas's theorem \Cref{lemma:lucas}, we now argue that if the slice $k$ is $(d,p)$-good, then for all $0 \leq i < d$,
\begin{align*}
    \binom{k-i}{d-i} \nequiv 0 \mod{p}
\end{align*}

Let $r = p^{\ell}$ be the smallest power of $p$ strictly greater than $d$. From \Cref{defn:good-slices}, we know that $k \equiv d \mod{p^{\ell}}$. Fix any $0 \leq i \leq (d-1)$. If we represent $(k-i)$ and $(d-i)$ in $p$-ary representation, then, we get,
\begin{align*}
    (d-i) \; = \; \sum_{j = 0}^{\ell-1} b_{j} p^{j} + \sum_{j = \ell}^{m} 0 p^{j}, \quad \quad \quad
    (k-i) \; = \; \sum_{j = 0}^{\ell-1} a_{j} p^{j} + \sum_{j = \ell}^{m} a_{j} p^{j}
\end{align*}
One can verify that $a_{j} = b_{j}$ for all $0\leq j \leq (\ell-2)$ and $a_{\ell - 1}  \geq b_{\ell - 1}$.\footnote{This is true by assumption for $i=0.$ The fact that it is also true when $i > 0$ follows from elementary properties of subtraction.} Thus by Lucas's theorem \Cref{lemma:lucas}, $\binom{k-i}{d-i} \nequiv 0 \mod{p}$. This concludes our proof that a degree strictly less than $d$ monomial can be expressed as a linear combination of monomials in $\mathcal{H}_{d}$.
\end{proof}

Note that \Cref{lem:dim-wilson} and \Cref{claim:homogeneous-spanning-positive-char} together imply \Cref{lem:basis}.

Now we are ready to prove the following distance lemma for degree-$d$ polynomials for good slices.\newline
The proof will be by a random restriction, which allows us to reduce to the case of \Cref{lemma:balanced-slice}. It random restriction sets a uniformly random set of $(n-2k)$ variables to $0$ (i.e. we are reducing from $\Boo^{n}_{k}$ to $\Boo^{2k}_{k}$). The main step is to argue that a non-zero polynomial on the $k^{th}$ slice continues to be a non-zero polynomial on the balanced slice (in a smaller dimension) with good enough probability. For this, we use $\mathcal{H}_{d}$ as a basis $\mathcal{P}_{d}(n,k,\mathbb{Z}_{q})$ for good slices.\\

\noindent
\begin{thmbox}
\begin{lemma}[Distance Lemma over good slices]\label{lemma:good-slice-positive-char}
There exists an absolute constant $\gamma > 0$ so that the following holds.
Fix a cyclic group $\mathbb{Z}_{q}$ where $q$ is a power of a prime number $p$. Fix a degree parameter $d \in \mathbb{N}$. For all natural numbers $n,k$ such that $1+d\le k^\gamma$ and $k\leq n/2$ and $k$ is $(d,p)$-good, the following holds.\newline
For every non-zero degree-$d$ polynomial $P(\mathbf{x}) \in \mathcal{P}_{d}(n,k,\mathbb{Z}_{q})$, 
\begin{align*}
    \Pr_{\mathbf{x} \sim U_{n,k}}[P(\mathbf{x}) \neq 0] \; \geq \; \alpha^{d}\paren{1  - \dfrac{1}{k^{\Omega(1)}}}, \text{ where } \alpha := \dfrac{k}{n}
\end{align*}
\end{lemma}
\end{thmbox}
\begin{proof}[Proof of \Cref{lemma:good-slice-positive-char}]
We start by describing the random process to reduce the problem from $k^{th}$ slice to the balanced slice in a smaller dimension cube.

\paragraph{\underline{Random process and the new polynomial}}Sample a random subset $T \subseteq [n]$ of size exactly $2k$ and set all the variables NOT in $T$ to $0$. Let $\Tilde{P}(y_{1},\ldots,y_{2k})$ be the resulting polynomial in $2k$ variables.\footnote{We identify the elements in $T$ with $[2k]$ in a canonical way.} Note that $\deg(\Tilde{P}) \leq \deg(P) = d$.\\
We will now argue that if $P(x_{1},\ldots,x_{n})$ is a non-zero degree-$d$ polynomial in $\mathcal{P}_{d}(n,k,\mathbb{Z}_{q})$, then $\Tilde{P}(y_{1},\ldots,y_{2k})$ is a non-zero degree-$d$ polynomial in $\mathcal{P}_{d}(2k,k,\mathbb{Z}_{q})$ with some good enough probability.\\
\begin{claim}\label{claim:non-zero-most-vars-killed}
Let $P(\mathbf{x}) \in \mathcal{P}_{d}(n,k,\mathbb{Z}_{q})$ be a non-zero polynomial. Then $\Tilde{P}$, as defined above, is a non-zero polynomial over $\Boo^{2k}_{k}$ with probability at least $(2k/n)^{d} \cdot (1 - (d^{2}/2k))$ over the randomness of set $T$.
\end{claim}

\begin{proof}[Proof of \Cref{claim:non-zero-most-vars-killed}]
Firstly, represent the polynomial $P(x_{1},\ldots,x_{n})$ as a unique linear combination of monomials in $\mathcal{H}_{d}$ using \Cref{claim:homogeneous-spanning-positive-char}. Let $\mathfrak{m}$ be a monomial in $\mathcal{H}_{d}$ which has a non-zero coefficient in the polynomial $P(\mathbf{x})$. Assume without loss of generality that $\mathfrak{m} = x_{1} x_{2} \cdots x_{d}$. The probability over the choice of $T$ that $\set{1,\ldots,d} \subset T$ is:
\begin{align*}
\displaystyle\binom{n-d}{2k-d}/\displaystyle\binom{n}{2k}
\end{align*}
Since $d \leq 2k \leq n$, we have the following inequality:
\begin{align*}
    \displaystyle\binom{n-d}{2k-d} \; \geq \; \displaystyle\binom{n}{2k} \cdot \paren{\dfrac{2k-d}{n-d}}^{d}
\end{align*}
Using the inequality $(1-x)^{m} \geq (1 - mx)$ and upper bounding $(n-d)^{d}$ by $n^{d}$, we get the following inequality:
\begin{align*}
    \displaystyle\binom{n-d}{2k-d} \; \geq \; \paren{\dfrac{2k}{n}}^{d} \cdot \paren{1 - \dfrac{d^{2}}{2k}}
\end{align*}
By \Cref{lem:basis}, the probability of $\Tilde{P}$ is a non-zero polynomial function over $\Boo^{2k}_{k}$ is at least the probability that the monomial $\mathfrak{m}$ has non-zero coefficient in $\Tilde{P}$, and this is at least $(2k/n)^{d} \cdot (1-(d^{2}/2k))$. This finishes the proof of \Cref{claim:non-zero-most-vars-killed}.
\end{proof}

Note that if $\Tilde{P}(\Tilde{\mathbf{a}}) \neq 0$ for some $\Tilde{\mathbf{a}} \in \Boo^{2k}_{k}$, then $P(\mathbf{a}) \neq 0$ where $\mathbf{a}$ is obtained from $\Tilde{\mathbf{a}}$ and fixing the coordinates not in $T$ to $0$. It is also easy to see that for a random choice of $T$, if $\Tilde{\mathbf{a}} \sim U_{2k,k}$, then the corresponding $\mathbf{a} \sim U_{n,k}$. Thus we get,
\begin{align*}
    \Pr_{\mathbf{x} \sim U_{n,k}}[P(\mathbf{x}) \neq 0] \; \geq & \; \Pr_{T}[\Tilde{P} \; \text{ doesn't vanish on } \; \Boo^{2k}_{k}]\\ &~~~~~\cdot ~\Pr_{\mathbf{y} \sim U_{2k,k}}[\Tilde{P}(\mathbf{y}) \neq 0~|~\Tilde{P} \; \text{ doesn't vanish on } \; \Boo^{2k}_{k}] .
\end{align*}
By using $d\le k^{\varepsilon'}$ for sufficiently small $\varepsilon'$ and applying the distance lemma for the balanced slices \Cref{lemma:balanced-slice} on $\Tilde{P} \in \mathcal{P}_{d}(2k,k,d)$ and \Cref{claim:non-zero-most-vars-killed}, we get,
\begin{align*}
    \Pr_{\mathbf{x} \sim U_{n,k}}[P(\mathbf{x}) \neq 0] \; \geq \; \frac{1}{2^d}\paren{1 - \dfrac{1}{k^{\Omega(1)}}} \cdot  \paren{\dfrac{2k}{n}}^{d} \cdot \paren{1 - \dfrac{d^{2}}{2k}} \;
    \geq \; \paren{\dfrac{k}{n}}^{d}\paren{1 - \dfrac{1}{k^{\Omega(1)}}}.
\end{align*}
 This finishes the proof of \Cref{lemma:good-slice-positive-char}.
\end{proof}

Now we will show how to reduce a non-good slice to a good slice by randomly fixing some $\bigO(d)$ many variables to $1$. We will prove the following lemma.\newline

\begin{lemma}[Reducing any slice to a good slice]
    \label{lem:all-slices} Fix a cyclic group $\mathbb{Z}_{q}$ where $q$ is a power of a prime number $p$, a degree parameter $d \in \mathbb{N}$. Let $n$ and $k\in [n^{1/3},n/2]$ be positive integers, and let $0\le c \le 2d \le k^{\varepsilon}$ for a sufficiently small constant $\varepsilon>0$. Let $\beta \in (0,(k/n)^d)$ be such that for every non-zero polynomial $Q({\bf x}) \in \cP_d(n-c,k-c,\Z_q)$, it holds that
         \begin{align*}
            \Pr_{{\bf x} \sim U_{(n-c),(k-c)}}[Q({\bf x}) \ne 0] \ge \beta.
    \end{align*}
     Then for every non-zero polynomial $P({\bf x}) \in \cP_d(n,k,\Z_q)$, it holds that
    \begin{align*}
        \Pr_{{\bf x} \sim U_{n,k}}[P({\bf x}) \ne 0] \ge \beta \paren{1 -  {\dfrac{c}{{n^{0.1}}}}}.
    \end{align*}
\end{lemma}

Together with~\Cref{lemma:good-slice-positive-char}, this completes the proof of \Cref{lemma:DLSZ-slice-finite-fields} as shown below. We recall the statement first.\\

\begin{thmbox}
\cyclicgroups*
\end{thmbox}

\begin{proof}[Proof of~\Cref{lemma:DLSZ-slice-finite-fields}]  
We first argue that we can assume that $k\ge n^{1/3}$. Otherwise,~\Cref{lemma:suboptimal-DLSZ-slice} suffices to give us the desired bound. Indeed, if $k \le n^{1/3}$, we have
$$\frac{{{n-2d}\choose k-d}}{{n\choose k}} \ge \paren{\frac{k-d}{n}}^d\paren{\frac{n-2k}{n}}^d \ge \paren{\frac{k}{n}}^d\paren{1-\frac{d^2}{k}}\paren{1-\frac{2kd}{n}}\ge \paren{\frac{k}{n}}^d \paren{1-\frac{1}{k^{\Omega(1)}}}.$$

Hence, for the rest of the proof, we will assume that $k\le n^{1/3}$. We show below that it suffices to show that for every slice $k > d$, there exists $c\in [0,2d]$ such that the slice $(k-c)$ is $(d,p)$-good (see \Cref{defn:good-slices}). Assuming this, the premise of \Cref{lem:all-slices} is true with 
\begin{align*}
    \beta = \paren{\frac{k-c}{n-c}}^d \, \paren{1-  \, \dfrac{1}{(k-c)^\gamma}},
\end{align*}
for some constant $\gamma \in (0,1)$, by the distance lemma for good slices (\Cref{lemma:good-slice-positive-char}). 

The conclusion of~\Cref{lem:all-slices} then implies that
    \begin{align*}
        \Pr_{{\bf x}\sim U_{n,k}}[P({\bf x}) \ne 0] \; \geq \; \bigg(\frac{k-c}{n-c}\bigg)^d\paren{1-\frac{1}{(k-c)^\gamma}}\paren{1-\frac{c}{n^{0.1}}} \ge\paren{\frac{k}{n}}^d \paren{1- \frac{1}{k^{\Omega(1)}}},
    \end{align*} using $d\le k^{\varepsilon}$ and $c\le 2d$.
    Hence, it only remains to show that there exists a $c\in [0,2d]$ such that the slice $(k-c)$ is $(d,p)$-good.

     \noindent
    Let $d=\sum_{j=0}^\ell b_j p^j$ and $k=\sum_{j=0}^m a_j p^j$ be the $p$-ary representations of $d$ and $k$ respectively, with $b_\ell, a_m >0$ and $m \ge \ell$ (since $k \ge d$). We first note that for $c_1 \in [0,p^\ell-1]$ such that $c_1 \equiv k-d \mod{p^\ell}$, we have $k_1 := k-c_1 \equiv d \mod{p^\ell}.$ Or equivalently for $k_1 = \sum_{j=0}^m u_j p^j$ in $p$-ary representation, we have
     \begin{align}
        \label{eqn:cond-1}
        u_j = b_j,\text{~for all~}j<\ell.
     \end{align}
     We now show that there exists a $c_2 \in [0,d]$ such that when $k_2 := k_1-c_2$ is expressed in its $p$-ary representation as $k_2 = \sum_{j=0}^{m} v_j p^j$, it holds that $v_\ell \ge b_\ell$ and $v_j = u_j$, for all $j<\ell$. We have two cases:
     \begin{itemize}
        \item {\bf Case 1: $u_\ell \ge b_\ell$.} In this case, we can take $c_2=0$ and $k_2=k_1$.
        \item {\bf Case 2: $u_\ell < b_\ell$.} We take $c_2=(u_\ell+1)p^\ell \le b_\ell\cdot p^\ell \le d$. In this case, the $p$-ary representation of $k_2$ as described above satisfies $v_j = u_j$ for all $j < \ell$ and $v_\ell = p-1.$
     \end{itemize}
     Taking $c=c_1+c_2$, we see that $k_2=k-c$ is indeed $(d,p)$-good, with $c \le 2d$.
     
     This finishes the proof of \Cref{lemma:DLSZ-slice-finite-fields}.
\end{proof}

We now prove \Cref{lem:all-slices}.

\begin{proof}[Proof of~\Cref{lem:all-slices}]
    Note that the proof is trivial for $c=0$. We will now show the lemma for $c\ge 1$. Consider the following way of sampling a point ${\bf x}$ from $U_{n,k}$. 
    
    \begin{enumerate}
        \item Initialize $S=\emptyset$.
        \item Choose $s \in [n]\setminus S$ uniformly at random, set $x_{s} = 1$, and add $s$ to $S$.
        \item Repeat Step 2 until $|S|=c$. 
        \item Set the co-ordinates of $\mathbf{x}$ not in $S$ according to the distribution $U_{n-c,k-c}$.
        \item Output ${\bf x}$.
    \end{enumerate}

    Let $P_1\in \cP_d(n-1,k-1,\Z_q),P_2\in \cP_d(n-2,k-2,\Z_q),\dots,P_c\in \cP_d(n-c,k-c,\Z_q)$ be the restrictions of the polynomial $P$ on the respective domains obtained by setting $x_{s_1}=1$, $x_{s_2}=1,\dots,x_{s_c}=1$ successively by the above random procedure, so that $S=\{s_1,s_2,\dots,s_c\} \in {[n] \choose c}$. 
    
    We claim that $P_1$ is non-zero on $\{0,1\}^{n-1}_{k-1}$, with high probability over the choice of $s_1$. To prove this, let us call an index $s\in [n]$ ``bad'' if for all ${\bf a}\in \{0,1\}^n_k$ such that $a_s=1$, we have that $P({\bf a})=0$. We observe that the probability of $P_1$ being entirely zero over $\{0,1\}^{n-1}_{k-1}$ is equal to the probability of a uniformly random $s \in [n]$ being bad. We show below that the number of such bad indices $s \in [n]$ is at most $\ell:=\lfloor{n/\sqrt{k}}\rfloor $. Towards a contradiction, suppose there are some $\ell$ bad indices $i_1,i_2,\dots,i_\ell \in [n]$. This means that if at least one of these co-ordinates takes value 1, $P({\bf x})$ evaluates to $0$. Thus the number of non-zeroes of $P$ in $\{0,1\}^n_k$ is upper bounded by the number of points that take the value $0$ on all these bad indices, i.e., ${n-\ell \choose k}$. However, by the weak distance lemma (\Cref{lemma:suboptimal-DLSZ-slice}), we know that the number of non-zeroes has to be at least ${n-2d \choose k-d}$. This yields a contradiction as we have ${n-\ell \choose k} < {n-2d \choose k-d}$ by the following claim. 
    
    \begin{claim}
        \label{clm:non-roots} 
        For every sufficiently large integer $n$ and arbitrary integers $k\in [n^{1/4},n/2]$, $d\in [1, k^{0.1}]$, and $\ell = \lfloor{n/\sqrt{k}}\rfloor$, we have
         $${n-\ell \choose k} < {n-2d \choose k-d}.$$
    \end{claim}
    \begin{proof}
        We have $${n-\ell \choose k} \le {n-2d \choose k}\cdot \bigg(1-\frac{k}{n}\bigg)^{\ell -2d} < {n-2d\choose k-d}\cdot n^d \cdot \bigg(1-\frac{k}{n}\bigg)^{\ell -2d}.$$ Hence, it suffices to show that $(1-\frac{k}{n})^{\ell-2d}\le \frac{1}{n^d}$. Using $\ell - 2d \ge n/(4\sqrt{k})$, indeed we have that $(1-\frac{k}{n})^{\ell-2d} \le (1-\frac{k}{n})^{n/(4\sqrt{k})} \le e^{-{\sqrt{k}/4}} \le \frac{1}{n^d}$, using $d\le k^{0.1}$ and $k \ge n^{1/4}$.
    \end{proof}
    By using the notation $P' \notequiv 0$ to denote that a polynomial $P'\in \cP_d(n',k',\Z_q)$ has at least one non-zero evaluation over the underlying slice,
    we have $$\Pr_{s_1 \sim [n]}[P_1 \notequiv 0] \ge 1-\frac{\ell}{n} \ge 1-\frac{1}{\sqrt{k}} \ge 1-\frac{1}{n^{1/6}}.$$

    By the same argument for $P_2,\dots,P_c$, we get with probability at least $1-\frac{1}{(n-1)^{1/6}}$ (over the choice of $S$), that $P_c$ has at least one non-zero evaluation over $\{0,1\}^{n-c}_{k-c}$. Here, we note that we will need to show that~\Cref{clm:non-roots} is also applicable if we replace $n$ and $k$ with $n-i$ and $k-i$ for every $i\in [c]$; this follows as $c \le 2d\le k^{\varepsilon}$, so $d\le (k-i)^{0.1}$ and $k-i \in [(n-i)^{1/4},(n-i)/2]$. Hence, at each step $i\in [1,c-1]$, conditioned on $P_i$ being non-zero over $\{0,1\}^{n-i}_{k-i}$, we get that $P_{i+1}$ is non-zero over $\{0,1\}^{n-i-1}_{k-i-1}$ with probability at least $1-\frac{1}{(n-i)^{1/6}} \ge 1-\frac{1}{n^{0.1}}$.

    Finally, applying the premise of the lemma statement (\Cref{lem:all-slices}) for $Q=P_c$, we get the desired bound. More formally, we obtain \\
    \begin{align*}   
        \Pr_{{\bf x}\sim U_{n,k}}[P({\bf x}) \ne 0] & \ge \Pr_{s_1 \sim [n]}[P_1 \notequiv 0] \cdot \Pr_{s_2 \sim [n]\setminus \{s_1\}}[P_2 \notequiv 0~\mid~P_1 \notequiv 0] \cdot \dots \\ &~~~~~~~\dots \cdot \Pr_{s_c \sim [n]\setminus \{s_1,\dots,s_{c-1}\}}[P_c \notequiv 0~|~P_{c-1} \notequiv 0] \cdot \Pr_{{\bf y}\sim U_{n-c,k-c}}[P_c({\bf y})\ne 0~\mid~P_c \notequiv 0]\\
        & \ge \bigg(1-\frac{1}{n^{0.1}}\bigg)^{c} \cdot \beta\\
        & \ge \beta \paren{1- \frac{c}{n^{0.1}}}.
    \end{align*}
    This concludes the proof of~\Cref{lem:all-slices}.
\end{proof}

\subsection{General Abelian Groups}\label{ssec:general-proof}
In this subsection, we use the distance lemma for slices over cyclic groups of prime power order \Cref{lemma:DLSZ-slice-finite-fields} to get a distance lemma for slices over arbitrary Abelian groups. We prove \Cref{thm:main} now.

\begin{proof}[Proof of \Cref{thm:main}]
By negating all the variables if necessary, we can assume that $k\leq n/2$. Let $P(\mathbf{x}) \in \mathcal{P}_{d}(n,k,G)$ be the following polynomial:
\begin{align*}
    P({\bf x}) = \sum_{\substack{I \subseteq [n] \\ |I| \leq d}} a_I \cdot x^I,
\end{align*}
where $x^{I} := \prod_{j \in I} x_{j}$. We first argue that it suffices to prove \Cref{thm:main} for \textbf{finite} Abelian groups. Then we will use the fundamental theorem of finite Abelian groups and our distance lemma for cyclic groups (\Cref{lemma:DLSZ-slice-finite-fields}) to finish the proof.

\paragraph{\underline{Reducing to finitely generated Abelian groups}}
Define $G'$ to be the subgroup of $G$ generated by the coefficients of $P$, i.e.
\begin{align*}
    G' =  \left\langle \setcond{a_I}{I \subseteq [n], |I| \leq d} \right\rangle
\end{align*}
Observe that we can treat the polynomial $P(\mathbf{x})$ as a degree-$d$ polynomial over the group $G'$ with a non-zero evaluation on the slice $\{0,1\}^n_k$, i.e. $P(\mathbf{x}) \in \mathcal{P}_{d}(n,k,G')$. Thus we have reduced to the case of finitely generated Abelian groups. Using the structure theorem of finitely generated Abelian groups, we know that $G'$ is isomorphic to 
\begin{align*}
    G' \; \cong \; \Z_{p_1^{\ell_1}} \times \dots \times \Z_{p_s^{\ell_s}}\times \Z^r
\end{align*}
for some prime numbers $p_1,\dots,p_s$, positive integers $\ell_1,\dots,\ell_s$ and $r\in \Z_{\ge 0}$.

\paragraph{\underline{Reducing to finite Abelian groups}}We now show that there exists a \emph{finite} Abelian group $G''$ and a polynomial $P'' \in \cP_d(n,k,G'')$ such that the following holds:
\begin{align*}
    \Pr_{{\bf x} \sim U_{n,k}}[P({\bf x}) \ne 0] = \Pr_{{\bf x} \sim U_{n,k}}[P''({\bf x}) \ne 0]
\end{align*}
If $r=0$, we can take $G'' = G'$ and we are done.\newline
Otherwise, let $M$ be a prime number greater than the absolute values of the last $r$ co-ordinates (using the isomorphism between $G$ and $\Z_{p_1^{\ell_1}} \times \dots \times \Z_{p_s^{\ell_s}}\times \Z^r $) of the evaluations $P({\bf x})$ for all ${\bf x} \in \{0,1\}^n_k$. Then, we take $G'' := \Z_{p_1^{\ell_1}} \times \dots \times \Z_{p_s^{\ell_s}}\times \Z_M^r $. Let $a_{I} = (a_{I}(1),\ldots,a_{I}(s+r)) \in G'$ be a coefficient of $P(\mathbf{x})$. Define the coefficient $a''_{I} \in G''$ as follows:
\begin{align*}
    a''_{I} \, := \, (a_I(1),a_I(2),\dots,a_I(s), a_I(s+1) \mod M, \dots, a_I(s+r) \mod M)
\end{align*}
Now define the polynomial $P''(\mathbf{x})$ as follows:
\begin{align*}
    P''({\bf x}) = \sum_{\substack{I \subseteq [n] \\ |I| \leq d}} a''_I \cdot x^I,
\end{align*}
where $x^{I} := \prod_{j \in I} x_{j}$. For every ${\bf x} \in \{0,1\}^n_k$.
we have that $P({\bf x}) = 0 \iff P''({\bf x}) = 0$, since by the definition of $M$, the last $r$ co-ordinates of $P({\bf x})$ can only take values strictly in between $-M$ and $M$. Thus we have reduced to finite Abelian group $G''$.

\paragraph{\underline{Cyclic groups of prime power order}}We will now argue that we can further reduce it to the case of cyclic groups of prime power order. For simplicity of notation, let the primes $p_{s+1} = \dots = p_{s+r} = M$, exponents $\ell_{s+1} = \dots = \ell_{s+r} = 1$. We thus have $G'' \cong \mathbb{Z}_{p_{1}^{\ell_{1}}} \times \cdots \times \mathbb{Z} _{p_{s+r}^{\ell_{s+r}}}$ and $P''({\bf x}) \in \cP_{d}(n,k,G'')$ is non-zero on $\{0,1\}^n_k$.
This means that there exists $j \in [s+r]$ such that the polynomial $P''(\mathbf{x})$ is a non-zero degree-$d$ polynomial on $\Boo^{n}_{k}$ over the cyclic group $\mathbb{Z}_{p_{j}^{\ell_{j}} }$.\\

Using \Cref{lemma:DLSZ-slice-finite-fields}, we know there exists a constant $\varepsilon$ such that $P''(\mathbf{x})$ is non-zero on at least $\alpha^{d}(1 - \bigO(1/k^{\varepsilon})) \cdot \binom{n}{k}$ points  over the cyclic group $\mathbb{Z}_{p_{j}^{\ell_{j}} }$ where $\alpha = k/n$. This implies that the polynomial $P''(\mathbf{x})$ is non-zero on at least $\alpha^{d}(1 - \bigO(1/k^{\varepsilon})) \cdot \binom{n}{k}$ points over the group $G''$. As we argued above, this in particular implies that the polynomial polynomial $P(\mathbf{x})$ is non-zero on at least $\alpha^{d}(1 - \bigO(1/k^{\varepsilon}))$ fraction of points on the slice $\Boo^{n}_{k}$ over the Abelian group $G$. This finishes the proof of \Cref{thm:main}.
\end{proof}

\section{Low-degree Functions Over Slices}\label{sec:junta}
In this section we will give a simple proof of a lemma of Filmus and Ihringer~\cite{Filmus-Ihringer19} following the proof idea of Nisan and Szegedy~\cite{Nisan-Szegedy}. We will first give a couple of definitions and set the notations for this section.\\

\noindent
For a function $f(x_{1},\ldots,x_{n})$ on the slice $\Boo^{n}_{k}$ with coefficients in $\mathbb{R}$, and for any two coordinates $i,j \in [n] \times [n]$, define $f^{(ij)}$ to be the function where we swap the $i^{th}$ variable with the $j^{th}$ variable in the function $f$.\newline
\begin{definition}[Influence]\label{defn:influence}
Let $f:\Boo^{n} \to \mathbb{R}$ be a function on the slice $\Boo^{n}_{k}$. For $(i,j) \in [n] \times [n]$, the $(i,j)^{th}$-influence of $f$, denoted by $\mathrm{Inf}_{ij}(f)$, is defined as,
\begin{align*}
    \mathrm{Inf}_{ij}(f) \; := \; \dfrac{1}{4} \; \Pr_{\mathbf{x} \sim U_{n,k}}[f(\mathbf{x}) \neq f^{(ij)}(\mathbf{x})]
\end{align*}
The total influence of $f$, denoted by $\mathrm{Inf}(f)$, is defined as,
\begin{align*}
    \mathrm{Inf}(f) \; := \; \dfrac{1}{n} \sum_{1 \leq i < j \leq n} \mathrm{Inf}_{ij}(f)
\end{align*}
\end{definition}
Note that if $i = j$ in the above definition, then $\mathrm{Inf}_{ii}(f) = 0$ and it does not contribute anything towards the total influence.\\

A key lemma in the proof of \cite{Filmus-Ihringer19} is a lower bound on every non-zero influence (see \cite[Lemma 3.1]{Filmus-Ihringer19}). They showed that there exists a constant $\alpha$ such that every non-zero influence of a degree-$d$ polynomial on the balanced slice is at least $\alpha^{d}$. The proof of this lemma in~\cite{Filmus-Ihringer19} uses analytic techniques such as the Log-Sobolev inequality on the Boolean slice~\cite{LeeYau} and the Hypercontractive inequality~\cite{DiaconisSaloffCoste}. Using our distance lemma for the balanced slices (\Cref{lemma:balanced-slice}), we can improve the lower bound to almost $1/2^{d}$ (which is easily seen to be tight up to constant factors). Note that the main result of this section only holds for degree $d \leq C \log n$ for some absolute constant $C > 0$, it suffices to use the simpler proof of \Cref{lemma:balanced-slice}. We state the lemma below.\\

\begin{thmbox}
\begin{lemma}[Lower bound on influences]\label{lemma:influence-lower-bound}
Let $f(x_{1},\ldots,x_{n})$ be a non-zero degree-$d$ function on the balanced slice $\Boo^{n}_{n/2}$. Then for every $(i,j) \in [n] \times [n]$ for which $\mathrm{Inf}_{ij}(f) > 0$, the following holds:
\begin{align*}
    \mathrm{Inf}_{ij}(f) \; \geq \;  \dfrac{1}{4} \cdot \dfrac{1}{2^{d}} \cdot \paren{1 - \dfrac{1}{n^{\varepsilon}}}
\end{align*}
for some absolute constant $\varepsilon > 0.$
\end{lemma}
\end{thmbox}

\begin{proof}
Fix some pair $(i,j) \in [n] \times [n]$ with $\mathrm{Inf}^{(ij)}(f) > 0$ and consider the polynomial $g_{ij}$ on the balanced slice, defined as follows: $g_{ij}(\mathbf{x}) := f(\mathbf{x}) - f^{(ij)}(\mathbf{x})$.\newline
Observe that since $f$ is a degree-$d$ polynomial, $f^{(ij)}$ is also a degree-$d$ polynomial, which means $g_{ij}$ is also a degree-$d$ polynomial on the balanced slice. Since the influence $\mathrm{Inf}_{ij}(f) > 0$, this means that $g_{ij}$ is non-zero on the slice $\Boo^{n}_{n/2}$. Now using our distance lemma on the balanced slice \Cref{lemma:balanced-slice},
\begin{gather*}
    \Pr_{\mathbf{x} \sim U_{n,n/2}}[f(\mathbf{x}) \neq f^{(ij)}(\mathbf{x})] \; = \; \Pr_{\mathbf{x} \sim U_{n,n/2}}[g_{ij}(\mathbf{x}) \neq 0] \; \geq \; \dfrac{1}{2^{d}} \cdot \paren{1 - \dfrac{1}{n^{\varepsilon}}} \\
    \Rightarrow \mathrm{Inf}_{ij}(f) \; \geq \;  \dfrac{1}{4} \cdot \dfrac{1}{2^{d}} \cdot \paren{1 - \dfrac{1}{n^{\varepsilon}}}
\end{gather*}
for some absolute constant $\varepsilon > 0.$
\end{proof}

Filmus and Ihringer \cite[Lemma 3.3]{Filmus-Ihringer19} use this lower bound on non-zero influences to get a bound on junta of degree-$d$ polynomials on the balanced slice. Using the above-mentioned improved lower bound on non-zero influence, we can also improve the bounds in \cite[Lemma 3.3]{Filmus-Ihringer19}.\\

\begin{lemma}\label{lem:improved-FI}
There exists an absolute constant $C > 0$ such that for all degree parameters $d \in \mathbb{N}$ such that $d \leq C \log n$, the following holds. Every degree-$d$ polynomial on the slice $\Boo^{n}_{n/2}$ is a $\eta(d)$-junta, where
\begin{align*}
    \eta(d) = \bigO(d \cdot 2^{d}).
\end{align*}
\end{lemma}
\begin{proof}
The proof is essentially the same proof as in \cite{Filmus-Ihringer19}, except for one inequality which can be improved using \Cref{lemma:influence-lower-bound}. Let $f(\mathbf{x}) \in \mathcal{P}_{d}(n,n/2,\mathbb{R})$ and let $G$ be a graph on the vertex set $[n]$ where $(i,j)$ is an edge if $\mathrm{Inf}_{ij}(f) \geq 1/2^{d} \cdot (1 - 1/n^{\varepsilon})$, where $\varepsilon$ is the absolute constant from \Cref{lemma:influence-lower-bound}. Let $M$ be a maximal matching in $G$. We now proceed similar to the proof in \cite[Lemma 3.3]{Filmus-Ihringer19} and we request the reader to refer \cite{Filmus-Ihringer19} as we just highlight the changes in the proof here.

Using \Cref{lemma:influence-lower-bound}, we get the following two inequalities upper and lower bounding the influence:
\begin{gather*}
    \paren{\dfrac{1}{4} \cdot \dfrac{1}{2^{d}} \cdot \paren{1 - \dfrac{1}{n^{\varepsilon}}} } \cdot \paren{1 - \dfrac{1}{n}} \cdot M \; \leq \; \mathrm{Inf}(f) \; \leq \; d \\
    \Rightarrow M \leq \bigO(d \cdot 2^{d}),
\end{gather*}
where we used the assumption $d \leq C \log n$ in upper bounding $1/n^{\varepsilon}$ by $\frac{1}{10} \cdot \frac{1}{2^{d}}$. Following the argument of~\cite{Filmus-Ihringer19}, this gives us that $f$ is a $2M$-junta, i.e., a $\bigO(d\cdot 2^d)$-junta.
\end{proof}

As already noted in the introduction, a stronger upper bound of $\eta(d) = \bigO(2^d)$ follows from the work of~\cite{Filmus-Ihringer19, CHS-20} (and can also be obtained by plugging \Cref{lemma:influence-lower-bound} in place of \cite[Lemma 3.3]{Filmus-Ihringer19} in the proof of~\cite{Filmus-Ihringer19}). The advantage here is the relatively simple proof following exactly the template of~\cite{Nisan-Szegedy}.

\section{Improved Bound for Linear Functions}
\label{sec:deg1}

In this section, we show how to obtain an improvement over our distance lemma (\Cref{thm:main}) for the case of linear polynomials, i.e., $d=1$. In particular, we will show the following.\\

\begin{thmbox}
\begin{theorem}
    \label{thm:deg1} Let $G$ be an arbitrary Abelian group, and $n\ge 8$ and $k \in [n-1]$ be positive integers. Then, for every polynomial $P({\bf x}) \in \cP_1(n,k,G)$ that is non-zero on $\{0,1\}^n_k$, we have
    $$\Pr_{{\bf x}\sim \{0,1\}^n_k}[P({\bf x}) \ne 0] \ge \frac{t-1}{n},$$ where $t = \min\{k,n-k\}$.
\end{theorem}
\end{thmbox}

This is an improvement over~\Cref{thm:main} as we have an additive term of $1/n$ as opposed to $1/n^\varepsilon$ for some constant $\varepsilon \in (0,1)$. In particular, in the regime $k\leq n^\delta$ for small enough $\delta \in (0,1)$, the above theorem gives a lower bound of $1/n^{1-\delta}$, while~\Cref{thm:main} fails to give anything non-trivial.

We note that this is also an improvement over the weak distance lemma shown by~\cite{ABPSS25-ECCC} (i.e.,~\Cref{lemma:suboptimal-DLSZ-slice}) as it gives a lower bound of $\frac{k(n-k)}{n(n-1)}$ which is less than $\frac{k-1}{n}$ for $k \ge \sqrt{n}+1$. In particular, taking $P({\bf x}) = x_1$, we see that our bound is almost tight. Furthermore, the additive term of $1/n$ in the above theorem cannot be avoided (at least for $k=n/2$ and up to a constant factor) as the polynomial $P({\bf x}) = x_1+x_2+1$ in $\cP_1(n,k=n/2,\Z_2)$ which is non-zero with probability $\frac{1}{2}-\frac{1}{2(n-1)}$.

We will now prove~\Cref{thm:deg1}.

\begin{proof}[Proof of~\Cref{thm:deg1}]
   Let $P({\bf x}) = a_1 x_1 + a_2 x_2 + \dots + a_n x_n + c$, where $a_1,\dots,a_n, c \in G$. We may (and we will) assume that $a_1,a_2,\dots,a_n$ are not all equal (abbreviated as NAE from now), as otherwise the polynomial always evaluates to a constant over $\{0,1\}^k$ since $\sum_{i\in [n]} x_i = k$ for all ${\bf x} \in  \{0,1\}^n_k$. Hence the desired probability is equal to 1 since we are guaranteed that $P({\bf x}) \ne 0$ for at least one point ${\bf x} \in \{0,1\}^n_k$.  We may further assume that $k\le n/2$, as otherwise, we can consider the evaluations of the polynomial $P(1-x_1,\dots,1-x_n)$ over ${\bf x} \in \{0,1\}^n_{n-k}$.  

    We divide the proof into two cases depending on whether $k$ divides $n$:
    
    {\bf Case 1: $k$ divides $n$.} That is, $n=mk$ for some integer $m \ge 2$. Let $M \in [n]^{k \times m}$ be the matrix formed by arranging $[n]$ according to a uniformly random permutation.  Denoting the $j$-th column of $M$ by $M_j$, let $\su(M_j) = \sum_{i\in M_j} a_i$. We note that
    \begin{align}\label{eqn:prob:nonzero}\Pr_{x \sim \{0,1\}^n_k}[P({\bf x}) \ne 0] = \Pr_M[\su(M_1) \ne -c],\end{align} as both LHS and RHS are essentially picking $k$ elements (without replacement) from $a_1,\dots,a_n$ uniformly at random and checking whether their sum is not equal to $-c$. We construct $M$ by the following random process:
    \begin{itemize}
        \item Partition $[n]$ into $k$ buckets, $B_1,B_2,\dots,B_k \subseteq [n]$, each of size $m$, uniformly at random. 
        \item Set the $i$-th row of $M$ to be a uniformly random permutation of $B_i$, for each $i\in [k]$ independently.   
    \end{itemize}
    We say that an index $i\in [k]$ is ``good'' if the elements $(a_{j})_{j\in B_i}$ are NAE and ``bad'' otherwise. Since $a_1,\dots,a_n$ are NAE, one might expect that there is at least one good index with high probability (over the randomness of the first step i.e., choosing the buckets $B_1,\dots,B_k$). That is, we give an upper bound on 
    $$\Pr_{B_1,\dots,B_k}[\text{all the indices in }[k] \text{ are bad}].$$

    Note that for the above probability to be positive, the number of times each $a_i$ appears in $(a_i)_{i\in [n]}$ must be a multiple of $m$. Thus we may assume that the multiset $(a_i)_{i\in [n]}$ is of the form: $b_1$ (taken $f_1m$ times), $b_2$ (taken $f_2m$ times), $\dots b_\ell$ (taken $f_\ell m$ times), where $\ell \ge 2$ and $b_1,\dots,b_\ell$ are mutually distinct elements of $G$ and $f_1,f_2,\dots,f_\ell \ge 1$.

    We handle the case of $\ell = m =2$ and either $f_1=1$ or $f_2=1$ separately by the following claim.

    \begin{claim}
        \label{clm:small-case}
        If $\ell=m=2$ and at least one of $f_1$ or $f_2$ is equal to 1, then $$\Pr_{{\bf x}\sim \{0,1\}^n_k}[P({\bf x})\ne 0]\ge \frac{1}{2}-\frac{1}{2(n-1)} \ge \frac{k-1}{n}.$$
    \end{claim}

    Hence, for the rest of the argument, we will assume that $m,\ell$ and $f_i$'s are such that the premise of~\Cref{clm:small-case} are not met. We then bound the bad probability as follows:
    \begin{claim}\label{clm:bad-prob}
        $$\Pr_{B_1,\dots,B_k}[\text{all the indices in }[k] \text{ are bad}] = \frac{{k \choose f_1,~f_2 ,~\dots,~f_\ell}}{{n \choose f_1m,~f_2 m,~\dots,~f_\ell m}} \le \frac{1}{n}.$$
    \end{claim}
    
    Given~\Cref{clm:bad-prob}, we conclude that with probability at least $1-1/n$, there exists at least one index $i\in [k]$ such that $(a_j)_{j\in B_i}$ are NAE; suppose $h_1 \ne h_2 \in (a_j)_{j\in B_i} \subseteq  G$ be such elements. Recall that we permute the elements corresponding to $B_i$ randomly to form the $i$-th row of the matrix $M$. As this is performed independently across the rows, we may fix an arbitrary permutation of all the rows except the $i$-th one and argue a lower bound on the probability of $P({\bf x})$ being non-zero by~\eqref{eqn:prob:nonzero}. Let $g_1,\dots,g_i,\dots,g_k \in G$ be the corresponding group elements in the first column, where $(g_j)_{j\ne i}$'s are some constants and $g_i$ is picked uniformly at random from the multiset $(a_j)_{j\in B_i}$. We have $\su(M_1) = g_1 + \dots + g_k$. Notice that the sum corresponding to $g_i = h_1$ and $g_i = h_2$ are not equal as $h_1 \ne h_2$. Hence,
    $$\Pr_{g_i}[g_1 + \dots + g_k \ne -c] \ge 1/m.$$ Therefore, we get $$\Pr_M[\su(M_1)\ne -c~|~\text{there exists a good index }i\in [k]]\ge 1/m = k/n.$$ Combining with~\Cref{clm:bad-prob}, this gives us that $\Pr_M[\su(M_1) \ne -c] \ge (k-1)/n$, thus finishing the proof of the theorem (when $k$ divides $n$).

    We now move to the case when $k$ does not divide $n$. 
    
    {\bf Case 2: $k$ does not divide $n$.} Suppose $n = mk +k'$ for some positive integers $m$ and $k' < k$. Since $k \le n/2$, this implies $m\ge 2$. Similar to the previous argument, we will analyze the probability of $P({\bf x})$ being non-zero using the group elements. In particular, we would like to lower bound the probability of $\sum_{i\in A} a_i \ne -c$, where $A \subseteq [n]$ is subset of size $k$ chosen uniformly at random. We will first sample $A$ by the following process. 
    \begin{itemize}
        \item[--] Choose a uniformly random subset $B$ of $[n]$ of size $mk$.
        \item[--] Choose a uniformly random subset $A$ of $B$ of size $k$. 
    \end{itemize} We claim that with probability at least $mk/n$, the elements $(a_i)_{i\in B}$ are NAE. More formally, we claim that $\Pr_{B \sim {[n]\choose mk}}[(a_i)_{i\in B}\text{~are NAE}] \ge mk/n$. To show this, suppose $g_1 \ne g_2 \in G$ be two distinct elements in $(a_i)_{i\in [n]}$. We have two cases.
    \begin{itemize}
        \item {\bf Case (i): There exists a $g\in G$ that occurs at least $mk$ times in $(a_i)_{i\in [n]}$.} Since $g_1 \ne g_2$, at least one of these two elements is not equal to $g$; say $g \ne g_1$ without loss of generality. We claim that $g_1 \in (a_i)_{i\in B}$ implies $(a_i)_{i\in B}$ are NAE. To see this, we note that $n-mk = k' < mk$, so $g$ must always appear in $(a_i)_{i\in B}$. Hence, $(a_i)_{i\in B}$ are NAE whenever $g_1 \in (a_i)_{i\in B}$. As $B$ is a uniformly random subset of size $mk$, this event (i.e., $g_1 \in (a_i)_{i\in B}$) happens with probability at least $mk/n$. 
        \item {\bf Case (ii): No such element exists.} In this case, $(a_i)_{i\in B}$ are NAE for all choices of $B \in {[n]\choose mk}$. 
    \end{itemize}
    Therefore, $\Pr_{B \sim {[n]\choose mk}}[(a_i)_{i\in B}\text{~are NAE}] \ge mk/n$. Conditioned on $B$ satisfying the condition that $(a_i)_{i\in B}\text{~are NAE}$, we note that $\sum_{i\in A} a_i \ne -c$ with probability at least $(k-1)/(mk)$ by using the fact that we have already established in \Cref{thm:deg1} for the case when $k$ divides $n$. Hence, the final probability is
    \begin{align*}
        \Pr_{A\sim {[n]\choose k}}\bigg[\sum_{i\in A}a_i \ne -c\bigg] & \ge \Pr_{B\sim {[n] \choose mk}}\bigg[(a_i)_{i\in B}\text{ are NAE}\bigg]\cdot \Pr_{A \sim {B \choose k}}\bigg[\sum_{i\in A} a_i \ne -c~\mid~(a_i)_{i\in B}\text{ are NAE}\bigg]\\
        & \ge \frac{mk}{n}\cdot \frac{k-1}{mk} = \frac{k-1}{n}.
    \end{align*}
\end{proof} 

We end with the proofs of~\Cref{clm:small-case} and~\Cref{clm:bad-prob}.

\begin{proof}[Proof of~\Cref{clm:small-case}]
    We have $k=n/2$. Without loss of generality, suppose $f_1=1$. Recall that the multiset $(a_i)_{i\in [n]}$ is equal to $b_1$ repeated $2f_1 = 2$ times and $b_2$ repeated $2f_2 = n-2$ times, for some $b_1\ne b_2 \in G$. We may further assume that $P({\bf x}) = a_1x_1 + a_2 x_2 + \dots + a_n x_n +c$ where $a_1=a_2=b_1$ and $a_3=a_4=\dots=a_n = b_2$. As $\sum_{i=1}^n x_i = k$ for all points ${\bf x} \in \{0,1\}^n_k$, we have  
    \begin{align*}P({\bf x}) & = b_1(x_1 + x_2) + b_2 (x_3 + x_4 + \dots + x_n) +c \\
    & = b_1(x_1 + x_2) + b_2 (k-x_1-x_2) + c \\
    & = {(b_1 - b_2)} (x_1 + x_2) + {kb_2 + c}.
    \end{align*}
    Let $b:=b_1-b_2$ and $b':=kb_2+c$. Then we have, $P({\bf x}) = \begin{cases}  b',\text{~if~}x_1=x_2=0,\\
    b + b',\text{~if~}x_1+x_2 =1,\\
    2b + b',\text{~otherwise.}
     \end{cases}$ 
     
     Since $b\ne 0$, we have the implication $$b+b' =0 \implies (b' \ne 0 \text{~and~} 2b+b' \ne 0).$$ 
     Hence, $$\Pr_{{\bf x}\sim \{0,1\}^n_{n/2}}[P({\bf x})\ne 0] \ge \min\bigg\{\Pr_{{\bf x}\sim \{0,1\}^n_{n/2}}[x_1 + x_2 =1], ~\Pr_{{\bf x}\sim \{0,1\}^n_{n/2}}[x_1 + x_2 \ne 1]\bigg \} = \frac{1}{2}-\frac{1}{2(n-1)}.$$
\end{proof}

Now, we prove~\Cref{clm:bad-prob}.

\begin{proof} [Proof of~\Cref{clm:bad-prob}]
    We consider two cases depending on the value of $m$.
    \begin{itemize}
        \item {\bf Case 1: $m=2$.} We have $n=2k$. We note that at least one of the two conditions below must be met:
        \begin{itemize}
            \item $\ell \ge 3$, or 
            \item $\ell = 2$ and $f_1,f_2 \ge 2$.
        \end{itemize}
        Regardless of which of the above two conditions is satisfied, we can always partition the multiset $(a_i)_{i\in [n]}$ into two sub(multi)sets $S_1 \subseteq G$ and $S_2 \subseteq G$, 
        such that $4\leq |S_1|\leq n/2$ 
        and for every $g_1 \in S_1$ and $g_2 \in S_2$, we have that $g_1 \ne g_2$. We now note that a necessary condition for an index $i\in [k]$ to be bad is $\{a_j|j\in B_i\}$ being a subset (as a multiset) of $S_1$ or $S_2$. Hence, the probability that all $i\in [k]$ are bad is at most ${{k\choose f}}/{{n \choose 2f}}$, where $f:=|S_1|/2 \in [2,k/2]$.  As ${{k\choose f}}/{{n \choose 2f}}$ is an increasing function of $f$ when $f \le k/2$, we can lower bound it by the value corresponding to $f=2$, i.e., $\frac{{k \choose 2}}{{n \choose 4}} = \frac{3}{(n-1)(n-3)} \le \frac{1}{n}$.
        \item {\bf Case 2: $m \ge 3$.} We note that the numerator of the fraction in~\Cref{clm:bad-prob} is $A:={k \choose f_1,~f_2,~\dots,~f_\ell} \ge k$ (as each $f_i \ge 1$) and the denominator is ${km \choose f_1 m,~f_2 m,~\dots,~f_\ell m} \ge {k \choose f_1,~f_2,~\dots,~f_\ell}^m = A^m$ by a simple counting argument. Hence, we have
        \begin{align*}
            \frac{{k \choose f_1,~f_2,~\dots,~f_\ell}}{{km \choose f_1 m,~f_2 m,~\dots,~f_\ell m}} \le \frac{A}{A^m}
            \le \frac{1}{k^{m-1}}
             \le \frac{1}{mk} = \frac{1}{n}.
        \end{align*}
        The last inequality follows using $k^{m-2} \ge m$ for $k,m \ge 3$ and for $k=2,m=4$.
    \end{itemize}
\end{proof}

\medskip

\printbibliography[
heading=bibintoc,
title={References}
] 

\addtocontents{toc}{\protect\setcounter{tocdepth}{1}}
\appendix

\section{Appendix}
\subsection{Proofs of \Cref{clm:eig-1} and \Cref{clm:D-conc}}\label{app:claims-simple-proof}
\begin{proof}[Proof of~\Cref{clm:eig-1}]
    By the definition of the weight function $w$ of the generators, we have
    \begin{align*}
        \mu'_\emptyset = \sum_{{\bf y}\in \{0,1\}^{2m}_e} w({\bf y})
         = \sum_{d=0}^m {2m \choose 2d} \cdot \frac{1}{2^m \cdot {m\choose d}}
         = \sum_{d=0}^m \frac{{2m \choose 2d}}{{m\choose d}^2} \cdot \frac{{m\choose d}}{2^m}
         & \le \max_{d \in [0..m]}\bigg(\frac{{2m \choose 2d}}{{m\choose d}^2}\bigg) \cdot \sum_{d=0}^m \frac{{m\choose d}}{2^m}\\
        & = \max_{d \in [0..m]}\bigg(\frac{{2m \choose 2d}}{{m\choose d}^2}\bigg).
    \end{align*}
    Now, for every $d\in [0..m]$, 
    $$\frac{{2m \choose 2d}}{{m\choose d}^2} = \frac{(2m)!d!^2(m-d)!^2}{(2d)!(2m-2d)!m!^2} = \frac{{2m \choose m}}{{2d \choose d}{2m-2d \choose m-d}}=O\bigg(\frac{2^{2m} }{\sqrt{m}}\cdot \frac{\sqrt{d}}{ 2^{2d}} \cdot \frac{\sqrt{m-d}}{2^{2m-2d}}\bigg)=O\bigg(\sqrt{\frac{d(m-d)}{m}}\bigg) \le O(\sqrt{m}),$$ where the last inequality uses the AM-GM inequality. Therefore, $0<\mu'_\emptyset \le O(\sqrt{m})$. 
\end{proof}

Finally, we prove~\Cref{clm:D-conc}.

\begin{proof}[Proof of~\Cref{clm:D-conc}]
    Letting $B:=\{d\in [0..m]~|~|2d-m|>\sqrt{50m\log m}\}$, we can explicitly express the probability as $$\Pr_{{\bf x}\sim \mathcal{D}}[||{\bf x}| - m| > \sqrt{50 m \log m}] = \sum_{d\in B}{2m \choose 2d}\cdot \frac{1}{2^m {m\choose d}} = \sum_{d\in B}\frac{{2m \choose 2d}}{{m\choose d}^2}\cdot \frac{{m\choose d}}{2^m} \le O(\sqrt{m})\cdot \sum_{d\in B}\frac{{m\choose d}}{2^m},$$ by using the bound $\frac{{2m \choose 2d}}{{m\choose d}^2} \le O(\sqrt{m})$ from the proof of~\Cref{clm:eig-1}. To bound the second factor $\sum_{d\in B} \frac{{m\choose d}}{2^m}$, we use a Chernoff bound for the sum of $m$ i.i.d.~copies of a uniformly random Boolean variable. In particular, we get $\sum_{d\in B} \frac{{m\choose d}}{2^m} \le O(1/m^3)$. Hence $\Pr_{{\bf x}\sim \mathcal{D}}[||x|-m|>\sqrt{50m \log m}] \le O(1/m^{2})$.
\end{proof}

\subsection{Proof of \Cref{claim:upper-bound-good-matching}}\label{app:good-matching}

\upperboundgoodmatching*

\begin{proof}[Proof of \Cref{claim:upper-bound-good-matching}]
Let us first how to describe a good matching, or in other words, how to generate a good matching. We first choose a subset $T \subseteq \mathbf{u}^{-1}\set{1}$ of size $0 \leq k \leq t$ which will be matched outside the set $\set{2,4,\ldots,2t}$. To satisfy the good matching condition, it enforces that for every $i \in T$, $\mathcal{M}(2i) \in \set{1,3,\ldots,2t-1} \setminus T$. We choose remaining $k$ vertices from $\mathbf{u}^{-1}\set{0}$ which will be matched outside $\set{1,3,\ldots,2t-1}$. This also enforces that $(t-k) \geq k \Rightarrow k \leq t/2$. Finally, we also account for the possible matching. This gives us:
\begin{equation}
    \Pr_{\mathcal{M}}[\mathcal{M} \; \text{ is a good matching}] \, = \, \sum_{k=0}^{t/2} \dfrac{\binom{t}{k} \binom{t-k}{k} (t-k)! (n/2-t+k)!}{(n/2)!}
\end{equation}
We upper bound $\binom{t}{k}, \binom{t-k}{k}$, and $(t-k)!$ by $t^{t}$ for all $0 \leq k \leq t/2$. Since $t \leq n^{\delta}$, we can upper bound $(n/2-t+k)$ by $(n/2-t/2)!$ for all $0 \leq k \leq t/2$. Using these we get,
\begin{align*}
   \Pr_{\mathcal{M}}[\mathcal{M} \; \text{ is a good matching}] \; \leq \; \dfrac{t^{3t}}{\binom{n/2}{t/2} (t/2)!}
\end{align*}
Employing the standard binomial estimate that $\binom{N}{K} \geq (N/K)^{K}$ and Stirling's approximation, we get,
\begin{align*}
    \Pr_{\mathcal{M}}[\mathcal{M} \; \text{ is a good matching}] \; \leq \; \dfrac{t^{3t}}{n^{t/2}} \; = \; \dfrac{1}{n^{\Omega(t)}}
\end{align*}

\noindent
Now we upper bound the probability that a random matching is a $t$-self good matching. Observe that if a matching $\mathcal{M}$ is a $t$-self good matching, then it implies there exists a subset $\widetilde{T} \subset [t]$ of size $t/2$ such that $\mathcal{M}(2i-1) = 2i$ for $i \in \widetilde{T}$. We can have an arbitrary matching in the remaining $(n/2-t/2)$ vertices in both $L$ and $R$. This gives us:
\begin{align*}
     \Pr_{\mathcal{M}}[\mathcal{M} \; \text{ is a $t$-self good matching}] \; \leq \; \dfrac{\binom{t}{t/2}(n/2-t/2)!}{(n/2)!} \; \leq \; \dfrac{1}{n^{\Omega(t)}}
\end{align*}
This finishes the proof of \Cref{claim:upper-bound-good-matching}.
\end{proof}

\subsection{Proof of \Cref{lem:dim-wilson}}\label{app:wilson-proof}

\degdslicedimension*

The proof is by induction on the parameter $t = (k-d)\cdot(n-k-d).$

The base case of the induction corresponds to the case when $t = 0$, i.e. $d = \min\{k,n-k\}.$ In this case, we want to show that any function $f:\{0,1\}^n_k \rightarrow \mathbb{Z}_q$ is a degree-$d$ polynomial. To show this, it suffices to show that any $\delta$-function on $\{0,1\}^n_k$ can be written as a polynomial of degree at most $d.$ 

Consider without loss of generality the $\delta$-function at the point $\mathbf{a} = 1^k 0^{n-k}.$ If $d = k,$ then the monomial $x_1\cdots x_d$ computes exactly this $\delta$-function. On the other hand, if $d = n-k,$ we can instead use the polyomial $(1-x_1)\cdots (1-x_d).$ In either case, we are done. This proves the base case, i.e. that $|\mathcal{P}_d(n,k,\mathbb{Z}_q)| \geq q^{\binom{n}{d}}$ in this case.

For the induction, assume that $d < \min\{k, n-k\}.$ We claim that 
\begin{equation}\label{eq:Wilson-indn}
|\mathcal{P}_d(n,k,\mathbb{Z}_q)| \geq |\mathcal{P}_d(n-1,k,\mathbb{Z}_q)\times \mathcal{P}_{d-1}(n-1,k-1,\mathbb{Z}_q)| 
\end{equation}
which immediately implies the inductive case using the induction hypothesis\footnote{Note that the induction hypothesis is applicable as either $k-d$ or $n-k-d$ drops by $1$ while the other remains the same.}, as we have
\begin{align*}
    |\mathcal{P}_d(n-1,k,\mathbb{Z}_q)\times \mathcal{P}_{d-1}(n-1,k-1,\mathbb{Z}_q)|  &\geq  |\mathcal{P}_d(n-1,k,\mathbb{Z}_q)|\cdot |\mathcal{P}_{d-1}(n-1,k-1,\mathbb{Z}_q)|\\ 
    &\geq  q^{\binom{n-1}{d}}\cdot q^{\binom{n-1}{d-1}} = q^{\binom{n}{d}}.
\end{align*}

To prove \Cref{eq:Wilson-indn}, we give an injection $\iota$ from the set $\mathcal{P}_d(n-1,k,\mathbb{Z}_q)\times \mathcal{P}_{d-1}(n-1,k-1,\mathbb{Z}_q)$ to the set $\mathcal{P}_d(n,k,\mathbb{Z}_q)$. For each function in $\mathcal{P}_d(n-1,k,\mathbb{Z}_q),$ we fix arbitrarily a polynomial of degree at most $d$ representing this function, and do the same for functions in $\mathcal{P}_{d-1}(n-1,k-1,\mathbb{Z}_q)$ with the degree parameter being $d-1.$

Let $(P,Q)\in \mathcal{P}_d(n-1,k,\mathbb{Z}_q)\times \mathcal{P}_{d-1}(n-1,k-1,\mathbb{Z}_q)$ be the chosen polynomial representations of a pair of functions in the corresponding sets. Define $\iota(P,Q)$ to be 
\[
R(x_1,\ldots, x_n) = P(x_1,\ldots, x_{n-1}) + x_n \cdot Q(x_1,\ldots, x_{n-1}).
\]
The map is injective because the function computed by $P$ (and hence the underlying polynomial which is fixed by the function) can be obtained by restricting $R$ to the points where $x_n = 0$. Further, the function $Q$ can be obtained by evaluating $R$ at the points where $x_n = 1$ and subtracting the value of the polynomial $P$ evaluated at the same point. 

This proves \Cref{eq:Wilson-indn} and concludes the inductive case.
\end{document}